\documentclass[pre,floatfix,notitlepage,a4paper]{revtex4}
\pdfoutput=1

\usepackage[colorlinks,pdftex]{hyperref}

\usepackage{amsmath}
\usepackage{amsthm}

\usepackage{amsbsy}
\usepackage{amsfonts}
\usepackage{amssymb}
\usepackage{enumerate}
\usepackage{epsfig}
\usepackage{textcomp}
\usepackage{wasysym}
\usepackage{multirow}
\usepackage{rotating}

\usepackage{subfigure}

\usepackage{algorithmic}
\usepackage{algorithm}

\usepackage{bbm} 

\newcommand{\bs}{\boldsymbol}

\newtheorem{theorem}{Theorem}[section]

\newtheorem{lemma}[theorem]{Lemma}
\newtheorem{proposition}[theorem]{Proposition}
\newtheorem{corollary}[theorem]{Corollary}
\theoremstyle{definition}
\newtheorem{definition}{Definition}[section]

\theoremstyle{remark}

\begin{document}

\title{\bf Isomorphisms in Multilayer Networks}
\author{Mikko Kivel\"a} 
\affiliation{Oxford Centre for Industrial and Applied Mathematics, Mathematical Institute, University of Oxford, Oxford OX2 6GG, UK}
\author{Mason A. Porter}
\affiliation{Oxford Centre for Industrial and Applied Mathematics, Mathematical Institute, University of Oxford, Oxford OX2 6GG, UK} 
\affiliation{CABDyN Complexity Centre, University of Oxford, Oxford OX1 1HP, UK}

\date{\today}

\begin{abstract}

We extend the concept of graph isomorphisms to multilayer networks with any number of ``aspects'' (i.e., types of layering). In developing this generalization, we identify multiple types of isomorphisms. For example, in multilayer networks with a single aspect, permuting vertex labels, layer labels, and both vertex labels and layer labels each yield different isomorphism relations between multilayer networks. Multilayer network isomorphisms lead naturally to defining isomorphisms in any of the numerous types of networks that can be represented as a multilayer network, and we thereby obtain isomorphisms for multiplex networks, temporal networks, networks with both of these features, and more. We reduce each of the multilayer network isomorphism problems to a graph isomorphism problem, where the size of the graph isomorphism problem grows linearly with the size of the multilayer network isomorphism problem. One can thus use software that has been developed to solve graph isomorphism problems as a practical means for solving multilayer network isomorphism problems. Our theory lays a foundation for extending many network analysis methods --- including motifs, graphlets, structural roles, and network alignment --- to any multilayer network.
\end{abstract}

\maketitle



\section{Introduction}

Network science has been very successful in investigations of a wide variety of applications in a diverse set of disciplines. In many situations, it is insightful to use a naive representation of a complex system as a simple, binary graph, which allows one to use the powerful methods and concepts from graph theory and linear algebra; and numerous advances have resulted from this perspective \cite{Newmanbook}.
As network science has matured and as ever more complicated data have become available, it has become increasingly important to develop tools to analyze more complicated graphical structures \cite{Kivela2014Multilayer,Boccaletti2014}. For example, many systems that were typically studied initially as ordinary, time-independent graphs are now often represented as time-dependent networks~\cite{Holme2012Temporal}, networks with multiple types of connections \cite{KyuMin2015Towards}, or interdependent networks \cite{Gao2012Networks}. 
Recently, a multilayer-network framework was developed to represent a large number of such networked systems \cite{Kivela2014Multilayer}, and the study of multilayer networks has rapidly become arguably the most prominent area of network science. It has achieved important results in a diverse set of fields, including disease dynamics \cite{bauch2015}, functional neuroscience \cite{bassett2016}, ecology \cite{pilosof2017}, international relations \cite{Cranmer:2014ut}, transportation \cite{gallotti2014anatomy}, and more.

With the additional freedom in representing a multilayer network, numerous ways to generalize network concepts have emerged \cite{Kivela2014Multilayer,Boccaletti2014}. The different definitions can arise from different modeling choices and assumptions, which also have often been implicit (rather than explicit) in many publications. To make sense of the multitude of terminology and develop systematic methods for studying multilayer networks, one needs to start from first principles and define the fundamental concepts that underlie the various methods and techniques from network analysis that one seeks to generalize. For example, exploring the fundamental question, ``How is a walk defined in multilayer networks?'', led to breakthroughs in generalizing concepts such as clustering coefficients \cite{Cozzo2013Clustering,mendes2016}, centrality measures \cite{DeDomenico2013Mathematical,DeDomenico2014Navigability,Sole2014Centrality}, and community structure \cite{Mucha2010Community,DeDomenico2015identifying,jeub2016} in multilayer networks. In this article, we answer another fundamental question: ``When are two multilayer networks equivalent structurally?'' by generalizing the concept of graph isomorphism to multilayer networks. 

Any attempt to generalize a method that relies on graph isomorphisms to multilayer networks also necessitates generalizing the concept of graph isomorphisms.
 Very recently, there has been work on methods relying on (some times implicit) generalizations of graph isomorphism---especially in the context of small subgraphs known as ``motifs'' \cite{Milo2002Network}---for many network types that can be represented as multilayer networks~\cite{Kovanen2011Temporal,Wehmuth2015Multiaspect,Taylor2007Network,Paulau2015Motif,Bentley2016Multilayer,batt2016}. Further, other network analysis tools, such as structural roles  \cite{Borgatti1992Notions,rossi2015} and network comparison methods \cite{Conte2004Thirty,Kelley2003Conserved,Prvzulj2007Biological,Rito2010Threshold,Ali2014Alignment}, are based on graph isomorphisms.

Defining isomorphisms for multilayer networks yields isomorphism relations for each of the wide variety of network types that can be expressed using a multilayer-network framework. For example, one obtains isomorphisms for multiplex networks (in which edges are colored), interconnected networks (in which vertices are colored), and temporal networks~\cite{Kivela2014Multilayer}.
Instead of defining isomorphisms and related methods and tools separately for each type of network, we develop a general theory and set of tools that can be used for any types of multilayer network \cite{pymnet}. With our contribution, we hope to avoid a confusing situation in the literature in which elementary concepts, terminology, tools, and theory are developed independently for the various special types of multilayer networks.

The rest of this paper is organized as follows.
In Section~\ref{sec:basics}, we introduce the basics concepts, lay out the ideas behind multilayer isomorphisms, and summarize the results of our article.
In Section~\ref{sec:defs}, we give the permutation-group formulation of multilayer network isomorphisms and enumerate some basic properties of multilayer network isomorphisms and related automorphism groups. 
In Section~\ref{sec:computation}, we show how to solve a multilayer network isomorphism problem computationally by reducing it to an isomorphism problem in a vertex-colored graph. This reduction allows one to use graph isomorphism software packages to solve the multilayer network isomorphism problem, and we use it to show that multilayer isomorphism problems are in the same computational complexity class as the graph isomorphism problem. 
We provide tools for producing the reductions as a part of a multilayer analysis software \cite{pymnet}.
In Section~\ref{sec:examples}, we give examples of how one can use multilayer network isomorphisms for multiplex networks, temporal networks, and interconnected networks. Finally, in Section~\ref{sec:conclusions}, we conclude and discuss future research directions.


\section{Basic Concepts and Summary of Results}
\label{sec:basics}

\subsection{Multilayer Networks}

In recent years, there has been a growing interest in generalizing the concept of graphs in various ways to study graphical objects that are better suited for representing specific real-world systems. This has allowed increasingly realistic investigations of complex networked systems, but it has also introduced mathematical constructions, jargon, and methodology that are specific to research in each type of system.  The rapid development of such jargon has been overwhelming, and it has sometimes led to confusion and inconsistencies in the literature \cite{Kivela2014Multilayer}.

To unify the rapidly exploding, disparate language (and disparate notation) and to bring together the multiple concepts of generalized networks that include layered graphical structures, the concept of a ``multilayer network'' was developed recently ~\cite{DeDomenico2013Mathematical,Kivela2014Multilayer}. Reference~\cite{Kivela2014Multilayer} includes a list of about 40 mathematical constructions that can be represented using the framework of multilayer networks. Most of these structures are variations either of graphs in which vertices are ``colored'' (\emph{i.e.}, vertex-colored graphs, see \ref{sec:vertexcoloredisom} and \ref{sec:vertexcolored}) or of graphs in which edges are ``colored'' (\emph{i.e.}, multiplex networks, see \ref{sec:multiplex}). Both types of coloring can also occur in the same system, various types of temporal networks admit a natural representation as a multilayer network~\cite{Kivela2014Multilayer}, and other types of complications can also arise. This new unified framework has opened the door for the development of very general, versatile network concepts and methods, and the study of multilayer networks has rapidly become arguably the most prominent area of network science. See \cite{Kivela2014Multilayer,Boccaletti2014} for reviews of progress in the study of multilayer networks.

\begin{figure}[!htp]
\centering
\includegraphics[width=0.4\linewidth]{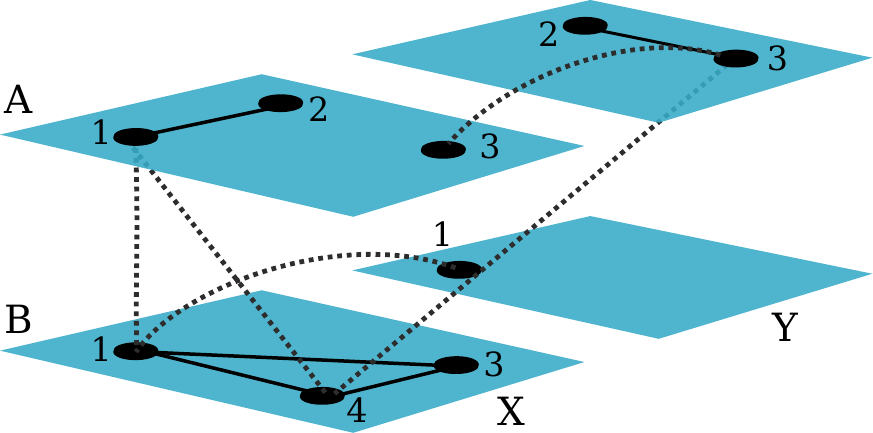}
\caption{Example of a multilayer network with two aspects. In this graphical structure, each entity can have an associated \emph{vertex-layer tuple} in one or more \emph{layers}, which are also organized into combinations of \emph{elementary layers}. In this example, each layer has either $A$ or $B$ as the first elementary layer, and it has either $X$ or $Y$ as the second elementary layer. There are thus $4$ layers in total, and entities can have an associated vertex-layer tuple in one, two, three, or all four layers. [This plot is inspired by a figure in \protect\cite{Kivela2014Multilayer}.]
}
\label{fig:mlayer}
\end{figure}

The formal definition of a multilayer network needs to be able to include the various layered network structures in the literature.  In Ref.~\cite{Kivela2014Multilayer}, we (and our collaborators) defined a multilayer network as a quadruplet $M=(V_M,E_M,V,\mathbf{L})$~\cite{Kivela2014Multilayer}: The set $V$ consists of the vertices of a network, just as is an ordinary graph. Each vertex resides in one or more uniquely-named \emph{layers} that are combinations of exactly $d$ \emph{elementary layers}, where each of these elementary layers corresponds to an ``aspect''. That is, each aspect is a different type of layering. For example, a social network that changes in time and includes social interactions over multiple communication channels has two aspects---one for time and the other for the type of social interaction---and so a layer represents one type of social interaction at a given time. The sequence $\mathbf{L}=\{ L_a \}_{a=1}^d$ consists of the sets of elementary layers for each of the $d$ aspects, and we use the symbol $\hat{L} = L_1 \times \dots \times L_d$ to denote the set of layers. Each vertex can be either present or absent in a layer, and we indicate the presence of a vertex by including its combination with the layer in the set $V_M \subseteq V \times L_1 \times \dots \times L_d$ of \emph{vertex-layer tuples}. Finally, we define the set $E_M \subseteq V_M \times V_M$ of edges between pairs of vertex-layer tuples as in ordinary graphs. See Fig. \ref{fig:mlayer} for an example of a multilayer network with two aspects.


\subsection{Isomorphisms in Graphs and Multilayer Networks}

A graph isomorphism formalizes the notion of two graphs having equivalent structures. 
The structure is what is left in a graph when one disregards vertex labels. That is, two graphs are isomorphic if one can transform one graph to the other by renaming the vertices in one of the graphs. Note that the edges do not have their own labels but they are determined by the vertex labels of the two endpoints, and those labels are also updated in the transformation.

To be able to give a mathematical definition of a graph isomorphism, we first define a \emph{vertex map} as a bijective function $\gamma: V \to V^\prime$ that relabels each vertex of the graph $G=(V,E)$ with another distinct label.  We use the following notation to relabel vertices of $G$ using $\gamma$: 
\begin{enumerate}
\item[(1)] $V^\gamma = \{ \gamma(v) \mid v \in V  \} $\,;
\item[(2)] $E^\gamma = \{ (\gamma(v),\gamma(u)) \mid (v,u) \in E \} $\,;
\item[(3)] $G^\gamma = (V^\gamma,E^\gamma)$\,.
\end{enumerate}
With this notation, two graphs $G$ and $G^\prime$ are \emph{isomorphic} if there exists $\gamma$ such that $G^\gamma=G^\prime$.

One can define isomorphisms for multilayer networks in very similar manner. The idea is again that two networks are equivalent structurally  if the vertices in one of them can be relabeled so that the first network is turned into exactly the second one. To do this, we need some additional (and slightly more cumbersome) notation:
\begin{enumerate}
\item[(1)] $V_M^\gamma = \{ (\gamma(v),\bs\alpha) \mid (v,\bs\alpha) \in V_M \} $\,;
\item[(2)] $E_M^\gamma = \{ ((\gamma(v),\bs\alpha),(\gamma(u),\bs\beta)) \mid ((v,\bs\alpha),(u,\bs\beta)) \in E_M \} $\,;
\item[(3)] $M^\gamma = (V_M^\gamma,E_M^\gamma,V^\gamma,\mathbf{L})$\,.
\end{enumerate}
Note that $(v,\bs\alpha) = (v,\alpha_1,\dots,\alpha_d)$, where $\bs\alpha$ is a vector of layers.
With the above definitions, we can now say that two multilayer networks $M$ and $M^\prime$ are \emph{vertex-isomorphic} when there exists $\gamma$ such that $M^\gamma = M^\prime$.

A vertex isomorphism is a natural extension of the standard graph isomorphism to multilayer networks, but it is not the only one. In a vertex isomorphism, one disregards only the vertex labels (but retains the layer labels) when comparing two multilayer networks. This choice is justifiable in some applications, but in others one might wish to also disregard the layer labeling. For example, one can map temporal networks into multilayer networks so that each time instance is a layer \cite{Kivela2014Multilayer}, and in this case two temporal networks are vertex-isomorphic if (1) the network has the same structure and order of structural changes and (2) the exact timings the structural changes are equal.
However, if one is interested only in the relative order of the changes that take place in the network, one needs to be able to also disregard the layer labels. To do this, one can proceed in very similar way as for vertices, as it requires a function to relabel the layers. Specifically, we say that a bijective function $\delta_a: L_a \to L_a^\prime$ is an \emph{elementary-layer map} that renames the elementary layers of a network.

\begin{figure}[!htp]
\includegraphics[width=0.7\linewidth]{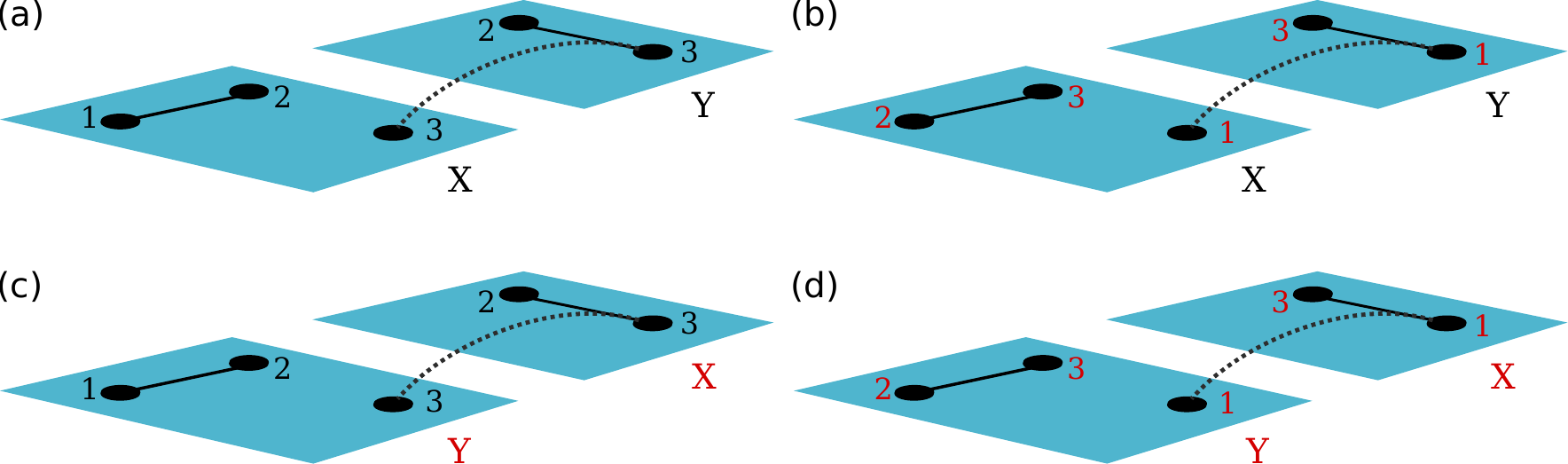}
\caption{Four examples of multilayer networks with one aspect: (a) $M_a$, (b) $M_b$, (c) $M_c$, and (d) $M_d$. The multilayer network $M_a$ is vertex-isomorphic to $M_b$, because there is a permutation $\gamma= (1 \, 2 \, 3)$ of vertex labels such that $M_a^\gamma=M_b$. We can thus write $M_a \cong_0 M_b$. The network $M_a$ is layer-isomorphic to $M_c$, and we write $M_a \cong_1 M_c$ because there is a permutation $\delta = (X \, Y)$ of layer labels such that $M_a^\delta=M_c$. The network $M_a$ is also vertex-layer isomorphic to $M_d$, and we write $M_a \cong_{0,1} M_d$ because there is a vertex-layer permutation $\zeta = (\gamma,\delta)$ such that $M_a^\zeta=M_d$. Note that $M_a$ is not vertex-isomorphic to $M_c$ or $M_d$, and it is not layer-isomorphic to $M_b$ or $M_d$. However, $M_a$ is vertex-layer isomorphic both to $M_b$ and to $M_c$.
}
\label{fig:isomex}
\end{figure}

One also may want to be able to relabel all of the elementary layers or only a subset of them in a multi-aspect multilayer network. We define a function $\bs\delta: \hat{L} \to \hat{L}^\prime$ that relabels all elementary layers, $\bs\delta(\bs\alpha)=(\delta_1(\alpha_1), \dots, \delta_d(\alpha_d))$, and call it a \emph{layer map}. A \emph{partial layer map} $\bs\delta_{j_1,\dots,j_k}$ is a layer map for which $\delta_a=1$ if $a \neq j_l$ for all $l$ and where $1$ is an identity map. That is, a partial layer map only relabels elementary layers that use some subset of all aspects. We say that these aspects $j_l$ are ``allowed'' to be mapped. 
We are now ready to define notation that formalizes the above ideas of how layer maps affect multilayer networks:
\begin{enumerate}
\item[(1)]$\mathbf{L}^{\bs\delta} =  \{ L_a^{\delta_a} \}_{a}^d $ and $L_a^{\delta_a} = \{ \delta_a(\alpha) \mid \alpha \in L_a  \} $\,;
\item[(2)]$V_M^{\bs\delta} = \{ (v,\bs\delta(\bs\alpha)) \mid (v,\bs\alpha) \in V_M \} $\,;
\item[(3)]$E_M^{\bs\delta} = \{ ((v,\bs\delta(\bs\alpha)),(u,\bs\delta(\bs\beta))) \mid ((v,\bs\alpha),(u,\bs\beta)) \in V_M \} $\,;
\item[(4)]$M^{\bs\delta} = (V_M^{\bs\delta},E_M^{\bs\delta},V,\mathbf{L}^{\bs\delta})$\,.
\end{enumerate}
We can now say that two multilayer networks are \emph{layer-isomorphic} when there exists a $\bs\delta$ such that  $M^{\bs\delta} = M^\prime$.
Because of the intrinsic complications in defining general multilayer networks with any arbitrary number of aspects, the above notation is a bit cumbersome. In Section \ref{next}, we will make the notation less cumbersome, at the cost of also making it less explicit.

We have defined isomorphisms related to relabeling either vertices or layers, but there is no reason why one cannot simultaneously do both of these. We thus define the \emph{vertex-layer map} $\bs\zeta=(\gamma,\delta_1, \dots, \delta_d)$ as a combination of a vertex map $\gamma$ and a layer map $\bs\delta$. A vertex-layer map acts on a multilayer network such that a vertex-map and layer-map act sequentially on the network: $M^{\zeta} = (M^\gamma )^{\bs\delta}$. Clearly, the order in which the vertices and layers are relabeled does not matter, and vertex maps and layer maps commute with each other, so $(M^\gamma )^{\bs\delta}=(M^{\bs\delta})^\gamma $. The vertex-layer maps can be used to define vertex-layer isomorphisms in the same way as one defines vertex isomorphisms and layer isomorphisms.

We now collect all of our definitions of multilayer-network isomorphisms.
\begin{definition}
Two multilayer networks $M$ and $M^\prime$ are
\begin{enumerate}
\item[(1)]\emph{vertex-isomorphic} if there is a vertex map $\gamma$ such that $M^\gamma=M^\prime$\,;
\item[(2)]\emph{layer-isomorphic}  if there is a layer map $\bs\delta$ such that $M^{\bs\delta}=M^\prime$\,;
\item[(3)]\emph{vertex-layer-isomorphic} if there is a vertex-layer map $\bs\zeta  = (\gamma,\bs\delta)$ such that $M^{\bs\zeta}=M^\prime$\,.
\end{enumerate}
Layer isomorphisms and vertex-layer isomorphisms are called \emph{partial isomorphisms} if the associated layer maps are partial layer maps.
\end{definition}

We use the notation $M\cong_{0}M^\prime$ to indicate that networks $M$ and $M^\prime$ are vertex-isomorphic. We indicate partial layer isomorphisms by listing the aspects that are allowed to be mapped (i.e., aspects that do not correspond to identity maps in the partial layer map) as subscripts. If the layer isomorphism is not partial, we list all of the aspects of the network. We use almost the same notation for vertex-layer isomorphisms, where the only difference is that we include $0$ as an additional subscript. For example, for a single-aspect multilayer network, $\cong_{1}$ denotes a layer isomorphism and $\cong_{0,1}$ denotes a vertex-layer isomorphism.  For partial layer isomorphisms and vertex-layer isomorphisms, we use a comma-separated list in the subscript to indicate the aspects that one is allowed to map. For example, $\cong_{2}$ is a partial layer isomorphism on aspect $2$, and $\cong_{0,1,3}$ signifies a vertex-layer isomorphism in which one is allowed to map aspects $1$ and $3$ but for which the layers in aspect $2$ (and in any aspects larger than $3$) are not allowed to change. We will explain the reason for this notation in Section \ref{next}.

We give examples of a vertex isomorphism, a layer isomorphism, and a vertex-layer isomorphism in Fig.~\ref{fig:isomex}.


\subsection{Summary of Results and Practical Implications}

\subsubsection{Applications of Isomorphisms}

The idea of a graph isomorphism is one of the central concepts in graph theory and network science, and it is an important underlying concept for many methods of network analysis---including motifs \cite{Milo2002Network}, graphlets \cite{Prvzulj2004Modeling,Prvzulj2007Biological}, graph matching~\cite{Conte2004Thirty,Kelley2003Conserved}, network comparisons \cite{Prvzulj2007Biological,Rito2010Threshold,Ali2014Alignment}, and structural roles \cite{Borgatti1992Notions}. Defining multilayer-network isomorphisms thus builds a foundation for future work by allowing generalization of all of these ideas for multilayer networks. Multilayer-network isomorphisms can also be used to define methods and concepts that are not intrinsic to graphs. For example, one can classify multilayer-network diagnostics and methods based on the types of multilayer isomorphisms under which they are invariant.  


\subsubsection{Applications of Automorphisms}

The modeling flexibility added by layered structures in networks has led to the discovery of qualitatively new phenomena (e.g., novel types of phase transitions) for processes such as disease spread and percolation \cite{Kivela2014Multilayer,domen2016,salehi2015}. It is interesting to examine how multilayer network architectures affect structural features such as the graph symmetries. One can study symmetries using automorphism groups of graphs, as these enumerate the ways in which vertices can be relabeled without changing a graph. We formulate the idea of graph automorphism groups for multilayer networks in Section \ref{sec:defs}, and we introduce a simplifying notation in which we think of vertices as a ``0th aspect''.  We show that combining maps of different aspects preserves all of the symmetries that are present in these aspects, but that completely new symmetries can result from combining these maps (see Proposition \ref{prop:autoprop}). For example, if a symmetry exists under vertex-isomorphism or layer isomorphism, it must also exist under vertex-layer isomorphism. However, a vertex-layer isomorphism can lead to symmetries that are not present under either vertex isomorphism or layer isomorphism.

The multilayer network automorphisms that we define in the present work generalize
notions of structural equivalence of vertices (or, more precisely, ``role equivalence'', ``role coloring'', or ``role assignment'') \cite{rossi2015,Borgatti1992Notions}.
Other related notions of structural equivalences have been defined in specific types of multilayer networks. For example, in social networks with multiple types of relations between vertices, one can study 
the ``block models'' that one obtains by considering different types of homomorphisms~\cite{Lorrain1971Structural,White1983Graph,Boyd1992Relational,boydbook}.
Additionally, in coupled-cell networks (which can have multiple types of edges and vertices), the automorphism groups and groupoids---which one obtains by relaxing the global condition for automorphisms---have a strong influence on the qualitative behavior of dynamical systems on such networks~\cite{Golubitsky2005Patterns,Golubitsky2006Nonlinear,Golubitsky2015Recent}.


\subsubsection{Aspect Permutations}

In our definition of multilayer networks, the elementary layers are ordered, and it is important to note that this is simply for bookkeeping purposes. Additionally, one can think of the vertices as elementary layers of a 0th aspect: from a structural point of view, the vertices are the same as other types of elementary layers. This is evident from the definition of multilayer networks, but it is far from evident in typical illustrations, in which vertices and layers are visualized, respectively, as points and planes. One can permute the order in which elementary layers are introduced, and isomorphism relations remain the same as long as the aspects in which the renamings are allowed are permuted accordingly (see Section \ref{sec:aspect_permutations}).

\subsubsection{Practical Computations}

For practical uses, it is important that the various types of multilayer isomorphisms can be computed in a simple and efficient way. It is a standard practice to solve this type of computational problem by reducing the problem to an isomorphism problem in (colored) graphs by constructing auxiliary graphs and then applying existing software tools for finding graph isomorphisms~\cite{Mckay2014Practical}. The auxiliary graphs can become complicated as the number $d$ of aspects grows, and a slightly different auxiliary graph construction procedure needs to be defined for all $2^d$ types of isomorphisms. In Section \ref{sec:computation}, we show how to construct such auxiliary graphs for general multilayer networks in a way that the size of the problem grows only linearly with the size of the multilayer network. This opens up a very straight forward and efficient way to apply our approach for practical data analysis of any kind of multilayer networks without requiring knowledge of reductions or explicit construction of auxiliary graphs.

\subsubsection{Application to Specific Network Types}

Most studies of multilayer networks usually consider specific types of multilayer networks rather than studying them in their most general form \cite{Kivela2014Multilayer}.
In Section \ref{sec:examples}, we show how the theory of multilayer isomorphisms can be applied to some of the most typical types of networks: multiplex networks, vertex-colored networks (i.e., networks of networks), and temporal networks. We also illustrate how the different implicit isomorphism definitions for temporal networks from the literature \cite{Kovanen2011Temporal} are related to our multilayer isomorphisms (see Section \ref{sec:temporal}). 
Our isomorphism definitions for multilayer networks are explicit, and anyone who is familiar with multilayer isomorphism can very easily transfer that knowledge to isomoprhisms in temporal networks.

One of the most prominent use of graph isomorphisms is motif analysis, in which all subgraphs of a network are grouped into isomorphism classes and the numbers of subgraphs in each class are examined \cite{Milo2002Network}. For both computational tractability and the ability to interpret the results of such calculations, such analysis typically relies on using a reasonably small number of isomorphism classes. This limits the sizes of subgraphs that are studied, as the number of isomorphism classes grows very rapidly as a function of number of vertices. Similarly, the number multilayer-network isomorphism classes grows very rapidly both as a function of the number of vertices and as a function of the number of layers. Consequently, the same limitations of motif analysis that apply to ordinary graphs also apply for multilayer networks. In Section \ref{sec:multiplex}, we examine the growth of the number of isomorphism classes in multiplex networks. This illustrates the type of compromise that one needs to make in the number of vertices and layers that can be considered in a subnetwork to ensure that the number of isomorphism classes is reasonable.


\section{Permutation Formulation and Properties of Multilayer Isomorphisms} \label{sec:defs}

We now show how to formulate the multilayer-network isomorphism problem in terms of permutation groups, and we give some elementary results for multilayer-network isomorphisms and related automorphism groups.


\subsection{Permutation Formulation of Multilayer Isomorphisms}\label{next}

We limit our attention (without loss of generality) to multilayer networks $\mathcal{M}$ in which each of the networks has the same set $V$ of vertices and same sets $\{ L_a \}_{a=1}^d$ of elementary layers.\footnote{For notational convenience in Section~\ref{sec:computation}, we assume that the vertices and layers can always be distinguished from each other. That is, we assume that the vertex set and the layer sets are distinct from each other and that any Cartesian product of the vertices and elementary layers are distinct from each other.}

We can now formulate the isomorphism theory using permutation groups. Vertex maps are permutations acting on the vertex set $V$, and elementary layer maps are permutations acting on elementary layer sets $L_a$. 
If we construe the group operation as the combination of two permutations, then all possible vertex maps form the symmetric group $S_V$, and all possible elementary layer maps for a given aspect $a$ form another symmetric group $S_{L_a}$
 (i.e., $\gamma \in S_V$, and $\delta_a \in S_{L_a}$).
The vertex-layer maps are given by a direct product of the symmetric groups of vertices and of elementary layers.

For notational convenience, we define the set of vertices to be the ``0th aspect'' (i.e., we define $L_0=V$). We also introduce the following notation for vertex-layer tuples: $\mathbf{v}=(v,\bs\alpha)$.  By convention, we define subscripts for vertex-layer tuples so that $v_0=v$ and $v_a=\alpha_a$ for $a > 0$, where $v \in V$ and $\bs\alpha \in \hat{L}$. It is also convenient to use $\mathbbm{1}_C$ to denote a group that consists of the identity permutation over elements of the set $C$.
Additionally, recalling that $\bs\zeta=(\gamma,\delta_1, \dots, \delta_d)$, it is convenient to use the notation ${\bf v}^{\bs\zeta}=\bs\zeta({\bf v})=(\gamma(v),\bs\delta(\bs\alpha))$ and $v_a^{\zeta_a}=\zeta_a(v_a)$.

We let $p \subseteq \{ 0,1,\dots,d \}$ (with $|p|\geq 1$) denote the set of aspects that can be permuted. Given $p$, we can then define permutation groups
\begin{equation}
	P_p = D_0^p \times \dots \times D_d^p\,,
\label{eq:ppdef}
\end{equation}
where $D_a^p=S_{L_a}$ if $a \in p$ and  $D_a^p=\mathbbm{1}_{L_a}$ if $a \notin p$. We denote the complementary set of aspects by $\overline{p}=\{ 0,1,\dots,d \} \setminus p$.

We obtain \emph{vertex permutations} for $p=\{ 0\} $, \emph{layer permutations} when $0 \notin p$, and \emph{vertex-layer permutations} when $0 \in p$ and $|p|>1$. Layer permutations or vertex-layer permutations are \emph{partial permutations} if there exists $a \in \{ 1, \dots d\}$ such that $a \notin p$.

We can now define multilayer-network isomorphisms for a set of multilayer networks $\mathcal{M}$.
\begin{definition}
Given a nonempty set $p$, the multilayer networks $M,M^\prime \in \mathcal{M}$ are \emph{$p$-isomorphic} if there exists $\bs\zeta \in P_p$ such that $M^{\bs\zeta}=M^\prime$. We write $M \cong_p M^\prime$.
\end{definition}

We denote the set of all isomorphic maps from $M$ to $M^\prime$ by $\mathrm{Iso}_p(M,M^\prime)=\{ \bs\zeta \in P_p: M^{\bs\zeta} = M^\prime \}$. Similarly, we use $\mathrm{Aut}_{p}(M)=\mathrm{Iso}_{p}(M,M)$ to denote the automorphism group of the multilayer network $M$.


\subsection{Basic Properties of Automorphism Groups}

In Eq.~(\ref{eq:ppdef}), we constructed the groups $P_p$ as direct products of symmetric groups and groups that contain only an identity element. The automorphism groups are subgroups of these groups: $\mathrm{Aut}_p(M) \leq P_p$. 
A permutation remains in the automorphism group even if we allow more aspects to be permuted (i.e., if the set $p$ is larger), and permutations that use only a given set of aspects are independent of permutations that use only other aspects. We formalize these insights in the following proposition.
\begin{proposition}
The following statements are true for all $M$ and $|p| > 0$:
\begin{enumerate}
\item[(1)] $\mathrm{Aut}_{p^\prime}(M) \leq \mathrm{Aut}_p(M)$ if $p^\prime \subseteq p$\,;
\item[(2)] $\mathrm{Aut}_{p_1}(M)\mathrm{Aut}_{p_2}(M) \leq \mathrm{Aut}_p(M)$ if $p_1,p_2 \subseteq p$, with $p_1 \cap p_2 = \emptyset$\,; 
\item[(3)] $\prod_i \bs\zeta^{(i)} = \bs{1} \implies \bs\zeta^{(i)}=\bs{1}$ for all $i$ if $\bs\zeta^{(i)} \in \mathrm{Aut}_{p_i}(M)$ and $p_i \cap p_j = \emptyset$ for all $i \neq j$\,.
\end{enumerate}
\label{prop:autoprop}
\end{proposition}
For a proof, see Section \ref{sec:proofs}.

It is important to observe in claim (2) of Proposition~\ref{prop:autoprop} that the subgroup relation can be proper even if $p=p_1 \cup p_2$. That is, the relationship $\mathrm{Aut}_{p_1}(M)\mathrm{Aut}_{p_2}(M) = \mathrm{Aut}_{p_1 \cup p_2}(M)$ is not always true, but one can combine permutations in $P_{p_1}$ and $P_{p_2}$ that are not in the automorphism groups $\mathrm{Aut}_{p_1}$ or $\mathrm{Aut}_{p_2}$ to obtain a permutation that is in $\mathrm{Aut}_{p_1 \cup p_2}(M)$. We give an example in Fig.~\ref{fig:permex}.

\begin{figure}[!htp]
\centering
\includegraphics[width=2.5in]{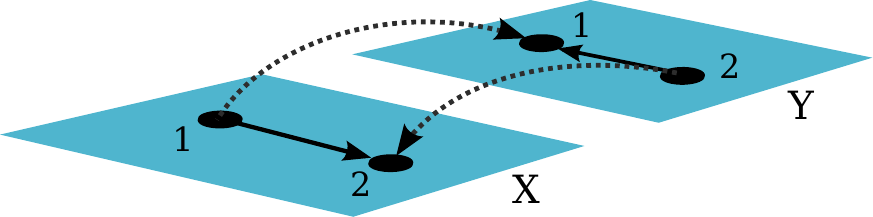}
\caption{An example demonstrating that one cannot always construct multilayer-network automorphism groups by combining smaller automorphism groups. 
That is, $\mathrm{Aut}_{p_1}(M)\mathrm{Aut}_{p_2}(M) \neq \mathrm{Aut}_{p_1 \cup p_2}(M)$ in this example.
In a directed multilayer network $M$ with edge set $E_M=\{ [(1,X),(2,X)],[(1,X),(1,Y)],[(2,Y),(1,Y)],[(2,Y),(2,X)] \}$, both the vertex automorphism group $\mathrm{Aut}_{\{ 0 \} }(M)$ and the layer automorphism group $\mathrm{Aut}_{\{ 1 \} }(M)$ are groups whose only permutation is the identity permutation, but the vertex-layer automorphism group $Aut_{\{ 0,1 \} }(M)$ has a permutation $((1 \, 2),(X \, Y))$ in addition to the identity permutation.
}
\label{fig:permex}
\end{figure}


\subsection{Aspect Permutations}
\label{sec:aspect_permutations}

In the definition of multilayer networks, the order in which one introduces different types of elementary layers (i.e., aspects) only matters for bookkeeping purposes. For example, for a system that is represented as a multilayer network with two aspects, $\mathcal{A}$ and $\mathcal{B}$, it does not matter if we assign index $1$ to aspect $\mathcal{A}$ and index $2$ to aspect $\mathcal{B}$ or index $1$ to aspect $\mathcal{B}$ and index $2$ to aspect $\mathcal{A}$. The isomorphisms of type $\cong_1$ and $\cong_2$ in the former case become the isomorphisms of type $\cong_2$ and $\cong_1$ in the latter case, and vice versa. 
Similar reasoning holds even if we consider the vertices to be a ``0th aspect'', as we did in Section~\ref{next}.

To formalize the above idea, we introduce the idea of \emph{aspect permutations} as permutations of indices of the aspects (including the 0th aspect). We then show that multilayer-network isomorphisms are invariant under aspect permutations as long as the indices in the set $p$ of aspects that are not restricted to identity maps are permuted accordingly.

\begin{definition}
Let $\sigma \in S_{\{ 0, \dots, d\} }$ be a permutation of aspect indices. We define an \emph{aspect permutation} of a multilayer network as 
$A_\sigma (M)=(V_M^\prime,E_M^\prime,V^\prime,{\bf L}^\prime)$, where
\begin{enumerate}
\item[(1)]$V_M^\prime =\{ (v_{\sigma^{-1}(0)},\dots, v_{\sigma^{-1}(d)}) \mid {\bf v} \in V_M \} $\,;
\item[(2)]$E_M^\prime = \{ [(v_{\sigma^{-1}(0)},\dots,v_{\sigma^{-1}(d)}),(u_{\sigma^{-1}(0)},\dots,u_{\sigma^{-1}(d)}) ] \mid ({\bf v},{\bf u}) \in E_M \} $\,;
\item[(3)]$V^\prime = L_{\sigma^{-1}(0)}$\,;
\item[(4)]${\bf L}^\prime = \{ L_{\sigma^{-1}(a)} \}_{a=1}^d $\,.
\end{enumerate}
\end{definition}

See Fig.~\ref{fig:transposition_example} for an example of an aspect permutation in a single-aspect multilayer network. For single-aspect multilayer networks, there is only one nontrivial aspect permutation operator, and we call the resulting multilayer network its \emph{aspect transpose}. Multilayer networks that are vertex-aligned~\cite{Kivela2014Multilayer} (i.e., networks for which $V_M=V_0 \times \dots \times V_d$) are often represented using adjacency tensors~\cite{DeDomenico2013Mathematical,Kivela2014Multilayer,DeDomenico2015Ranking}.
In this case, aspect permutations of multilayer networks become permutations of tensors indices~\cite{Kolda2009Tensor,Comon2008Symmetric,Pan2014Tensor} in the tensor representation.
Note that aspect permutation is a meaningful operation even for undirected multilayer networks, and it is different from the transpose operator, which reverses the orientations of the edges.

\begin{figure}[!htp]
\begin{minipage}{0.7\linewidth}
\centering
$\vcenter{\hbox{\includegraphics[width=0.4\linewidth]{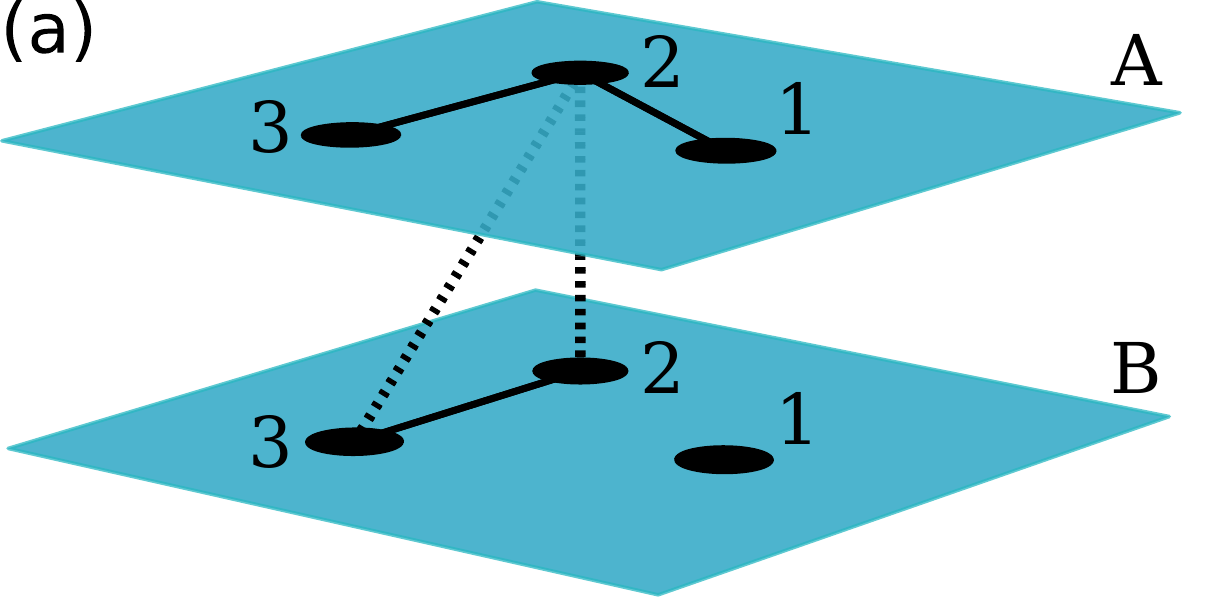}}}$
\hspace*{0.1\linewidth}
$\vcenter{\hbox{\includegraphics[width=0.4\linewidth]{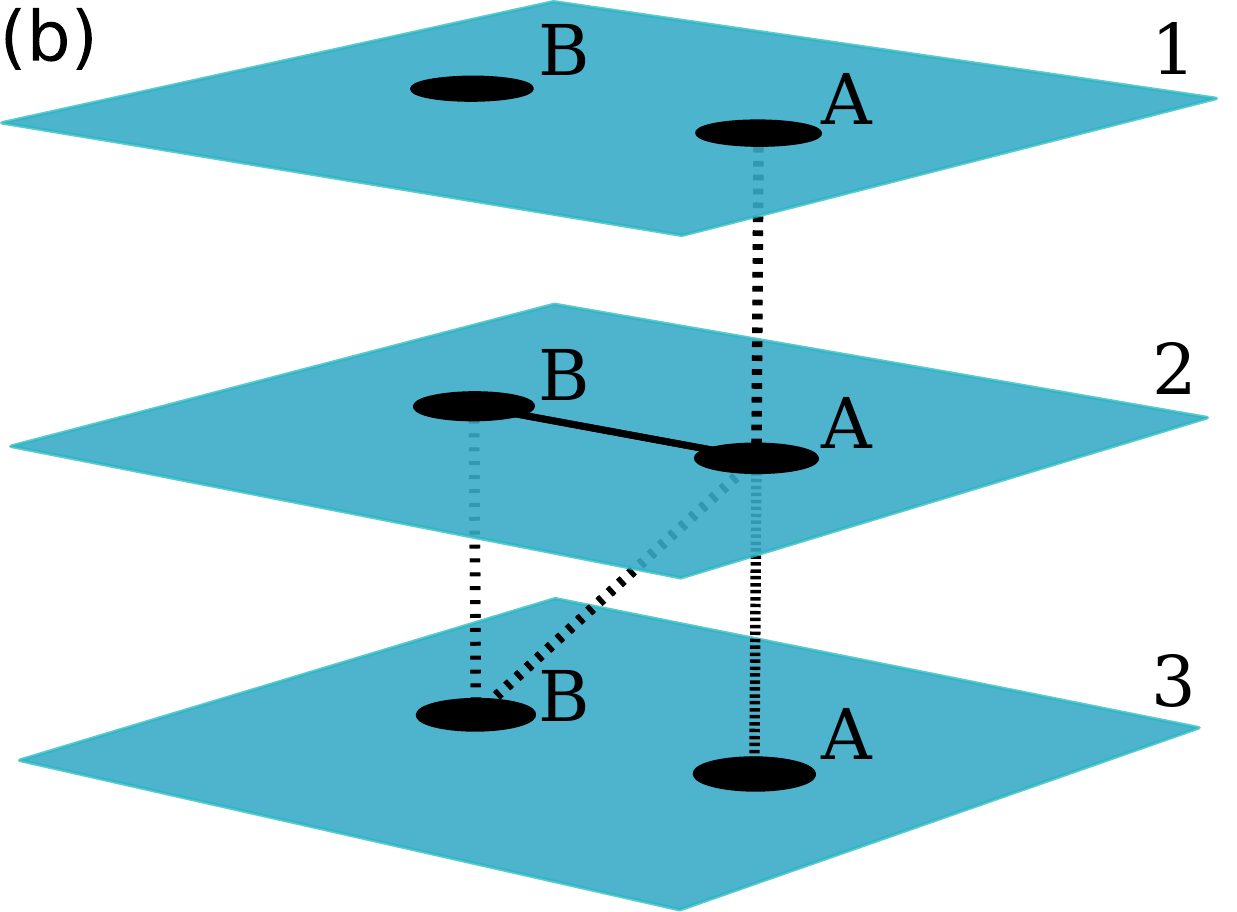}}}$
\end{minipage}
\caption{Two single-aspect multilayer networks, (a) $M_a$ and (b) $M_b$, that are aspect-transposes of each other: $M_a=A_{(0 1)}(M_b)$ and $M_b=A_{(0 1)}(M_a)$. The transposition operation preserves multilayer isomorphisms in the sense that a third multilayer network $M_c$ is vertex-isomorphic (respectively, layer-isomorphic) to $M_a$ if and only if $A_{(0 1)}(M_c)$ is layer-isomorphic (respectively, vertex-isomorphic) to $M_b$. Similarly, $M_c$ is vertex-isomorphic (respectively, layer-isomorphic) to $M_b$ if and only if $A_{(0 1)}(M_c)$ is layer-isomorphic (respectively, vertex-isomorphic) to $M_a$; and $M_c$ is vertex-layer-isomorphic to $M_b$ (respectively, $M_a$) if and only if $A_{(0 1)}(M_c)$ is vertex-layer-isomorphic to $M_a$ (respectively, $M_b$).
Illustrations produced using \protect\cite{pymnet}.
}
\label{fig:transposition_example}
\end{figure}

Aspect permutations preserve the sets of isomorphisms as long as the indices in the isomorphism permutations are also permuted accordingly. 
\begin{proposition}
The relation 
\begin{equation}\label{aspectpermrelation}
	\mathrm{Iso}_{p}(M,M^\prime) = I_{\sigma^{-1}} [\mathrm{Iso}_{p^{\sigma}}(A_\sigma(M),A_\sigma(M^\prime))]
\end{equation}	
holds, where $I_{\sigma^{-1}}({\bs\zeta})=(\zeta_{\sigma(0)}, \dots, \zeta_{\sigma(d)})$ is an operation that permutes the order of elements in a tuple according to the permutation $\sigma^{-1}$ and $p^\sigma$ is a set in which each element of $p$ is permuted according to the permutation $\sigma$.
\label{proposition:aspectpermutation}
\end{proposition}
For a proof, see Section \ref{sec:proofs}.


\section{Solving Multilayer Isomorphism Problems} \label{sec:computation}

To take full advantage of the theory of isomorphisms in multilayer networks, one needs efficient computational methods for finding isomorphisms between a pair of multilayer networks. 
One can proceed on a case-by-case basis for various types of networks, such as temporal networks~\cite{Kovanen2011Temporal}, using standard techniques from the graph-isomorphism literature~\cite{Mckay2014Practical}. We will now use the same techniques to show how to reduce all of the multilayer-network isomorphism problems to vertex-colored-graph isomorphism problems. 
This reduction allows one to solve any kind of isomorphism problem for any type of multilayer network without the need to come up with and prove the correctness of a new reduction technique.

In the reductions that we define, the size of the vertex-colored-graph isomorphism problem is a linear function of the size of the multilayer-network isomorphism problem and thus yields practical ways of solving multilayer-network isomorphism problems. 
We also use these reductions to show that solving multilayer-network isomorphism problems is in the same complexity class as ordinary graph isomorphism problems. This is unsurprising, as many generalized graph isomorphism problems are known to be equivalent~\cite{Zemlyachenko1985Graph}, including ones that involve the very general relational structures defined in Ref.~\cite{Miller1977Graph}. Another valid approach for our argument would be to reduce a multilayer-network isomorphism problem to other structures (e.g., to a $k$-uniform hypergraph~\cite{Codenotti2011Testing}), but the reduction to a vertex-colored graph yields practical benefits in terms of the ability to directly use software that is designed to solve isomorphism problems.


\subsection{Isomorphisms in Vertex-Colored Graphs} \label{sec:vertexcoloredisom}

A \emph{vertex-colored graph} $G_c=(V_c,E_c,\pi,C)$ is an extension of a graph $(V_c,E_c)$ with a surjective map $\pi: V \to C$ that assigns a color to each vertex. We define a vertex map as a bijective map $\gamma: V_c \to V'_c$ and introduce the following notation: $V^\gamma_c = \{ \gamma(v) \mid v \in V_c  \} $, $E^\gamma_c = \{ (\gamma(v),\gamma(u)) \mid (v,u) \in E_c\}$, $\pi^\gamma(v) = \pi(\gamma^{-1}(v))$, and $G^\gamma_c =(V^\gamma_c,E^\gamma_c,\pi^\gamma,C)$. Two vertex-colored graphs, $G_c$ and $G_c^\prime$, are \emph{isomorphic} if there is a vertex map $\gamma$ such that $G_c^\gamma = G_c^\prime$, and we then write $G_c \cong G_c^\prime$.

For the purposes of isomorphisms, we can---without loss of generality---limit our attention to graphs with the vertex set $V=\{1, \dots , n \}$, where $n$ is the number of vertices in the graph. This allows us to phrase the graph isomorphism problem in terms of permutations (similar to Section~\ref{next}). The bijective map $\gamma$ in the definition of a graph isomorphism is again a permutation that acts on the set $V$ of vertices, and the permutations form the symmetric group $S_V$.

The vertex-colored-graph isomorphism problem is a well-studied computational problem, and several algorithms and accompanying software packages are available for solving it~\cite{Miller1977Graph,Zemlyachenko1985Graph,Mckay2014Practical}.


\subsection{The Reduction}

\begin{figure*}[!htp]
\centering
\includegraphics[width=0.8\linewidth]{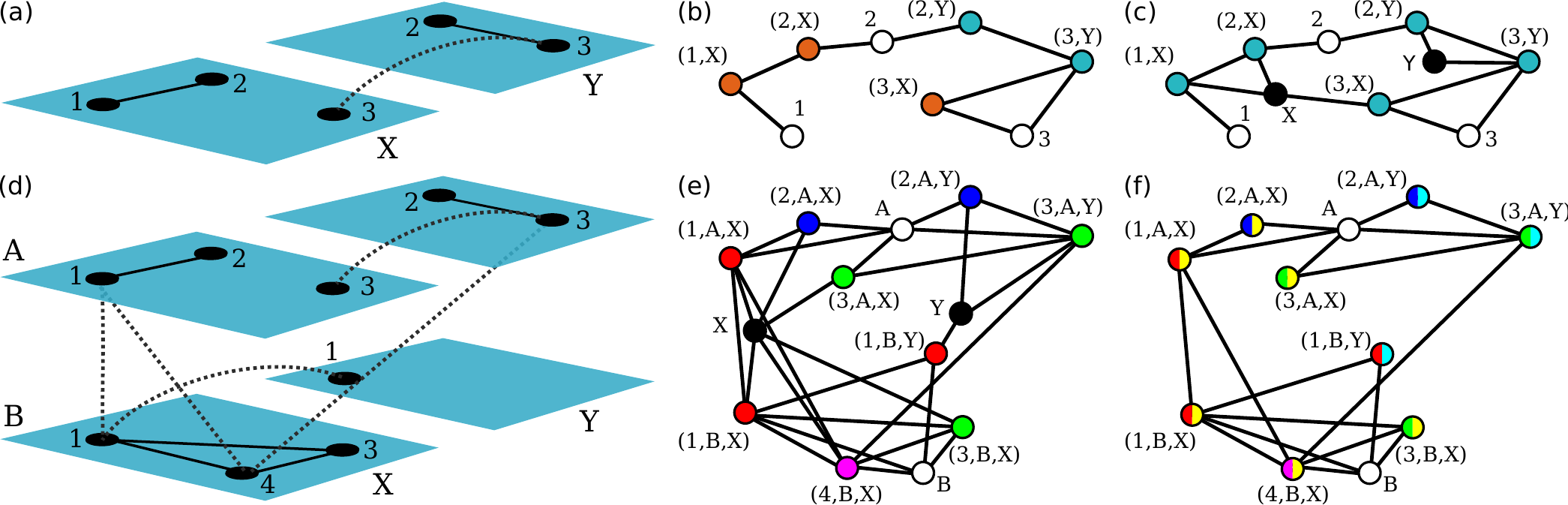}
\caption{Example of a function $f_p$ that maps multilayer networks to vertex-colored graphs. (a) A multilayer network $M_1$ with a single aspect, two layers, and three vertices. (b) The vertex-colored graph $f_{\{0\}}(M_1)$. One can use the mapping $f_{\{0\}}$ to find vertex isomorphisms in the multilayer network $M_1$. In other words, permutations of vertex labels are allowed, but permutations of layer labels are not allowed. (c) The vertex-colored network $f_{\{0,1\}}(M_1)$. One can use the mapping $f_{\{0,1\}}$ to find vertex-layer isomorphisms in the multilayer network $M_1$. In other words, both vertex labels and layer labels are both allowed to be permuted. (d) A multilayer network $M_2$ with two aspects, two layers in each aspect, and four vertices. (e) The vertex-colored graph $f_{\{1,2\}}(M_2)$. One can use the mapping $f_{\{1,2\}}$ to find layer isomorphisms in the multilayer network $M_2$. 
In other words, permutations of layer labels are allowed in each aspect, but permutations of vertex labels are not allowed. (f) The vertex-colored graph $f_{\{1\}}(M_2)$. One can use the mapping $f_{\{1\}}$ to find partial layer isomorphisms in the multilayer network $M_2$. In other words, permutations of layer labels are allowed only in the first aspect, and permutations of vertex labels or layer labels are not allowed in the second aspect.
}
\label{fig:vertexisom}
\end{figure*}

The idea behind our reduction of multilayer-network isomorphism problems to the isomorphism problem in vertex-colored graphs is that we define an injective function $f_p$ such that two multilayer networks $M$ and $M^\prime$ are isomorphic with a permutation from $P_p$ if and only if $f_p(M)$ and $f_p(M^\prime)$ are isomorphic vertex-colored graphs. In this reduction, it is useful to consider the concept of an \emph{underlying graph} $G_M=(V_M,E_M)$ of a multilayer network \cite{Kivela2014Multilayer}.  For two multilayer networks to be isomorphic, their underlying graphs need to be isomorphic. However, this is not a sufficient condition, because it allows (1) permutations in aspects that are not included in $p$ and (2) permutations that occur in each layer independently of permutations that occur in other layers. Consider, for example, the multilayer network $M_a$ in Fig.~\ref{fig:isomex} and the network $M_a^\prime$ that one obtains by swapping vertex labels $2$ and $3$ in layer $X$ but not in layer $Y$. The underlying graphs $G_{M_a}$ and $G_{M_a^\prime}$ are then isomorphic even though there is no vertex-layer isomorphism between the two associated multilayer networks.

We address the first issue above by coloring the vertices in the underlying graph so that its vertices, which correspond to vertex-layer tuples in the associated multilayer network, that are not allowed to be swapped are assigned different colors from ones that can be swapped. For example, for a vertex isomorphism in a single-aspect multilayer network, we color the vertices of the underlying graph according to the identity of their layers (i.e., by using a different color for each layer). We address the second issue above by gluing together vertex-layer tuples that share a vertex or an elementary layer by using auxiliary vertices. For example, for a vertex isomorphism in a single-aspect multilayer network, we add an auxiliary vertex in the underlying graph for each vertex $v \in V$ in the multilayer network, and we connect the auxiliary vertex to vertices in the underlying graph that correspond to $v$.  This restricts the possible permutations: for each layer, one needs to permute the vertex labels in the same way. 
See Fig.~\ref{fig:vertexisom} for an example of our reduction procedure.

We define the reduction function $f_p$ for general $M$ and $p$ as follows.

\begin{definition}
We construct the reduction from multilayer networks to vertex-colored graphs $f_p:\mathcal{M} \to \mathcal{G}_c$ such that $f_p\left((V_M,E_M,V,\bs{L})\right)=(V_G,E_G,C,\pi)$ using
\begin{enumerate}
\item[(1)]$V_G=V_M \cup V_0$, where the auxiliary vertex set $V_0= \bigcup_{a \in p} L_a$\,; 
\item[(2)]$E_G= E_M \cup E_0$, where $E_0=\{ (v_a,{\bf v}) \mid {\bf v} \in V_M, a \in p \}$\,; 
\item[(3)]$C=p \cup L_{\overline{p}_1} \times \dots \times L_{\overline{p}_m}$\,;
\item[(4)]$\pi(v_g)=a$ if $v_g \in L_a$ and  $\pi(v_g)=(v_{\overline{p}_1} \dots v_{\overline{p}_m})$ if $v_g={\bf v} \in V_M$\,.
\end{enumerate}
\label{definition:fp}
\end{definition}

In addition to the reduction function $f_p$ that we need to solve the decision problem of two multilayer networks being isomorphic, we would like to be able to explicitly construct the permutations that we need to map a multilayer network to an isomorphic multilayer network. That is, we need a mapping between the permutations in multilayer networks and permutations in vertex-colored graphs. 
We define this map as follows.
\begin{definition}
Given a multilayer network $M$, we define the function $g_p$ from the permutations $P_p$ to permutations of vertex-colored graphs so that $v_g^{g_p(\bs\zeta)}=v_g^{\bs\zeta}$ if $v_g \in V_M$ and $v_g^{g_p(\bs\zeta)}=v_g^{\zeta_a}$ if $v_g \in L_a$ for any $\bs\zeta \in P_p$.
\end{definition}

The following theorem allows us to use $f_p$ and $g_p$ for the purpose of solving multilayer network isomorphism problems using an oracle for vertex-colored graph isomorphism.
\begin{theorem} \label{theorem:mlayerisom}
$\mathrm{Iso}_p(M,M^\prime)=g^{-1}_p[\mathrm{Iso}(f_p(M),f_p(M^\prime))]$
\end{theorem}
For a proof see Section \ref{sec:proof_reduction}.

From Theorem \ref{theorem:mlayerisom}, it follows that one can also solve multilayer network isomorphism problems using the reduction to vertex-colored graphs that we have introduced. For example, one can use this reduction to determine if two multilayer networks are isomorphic, to define complete invariants for isomorphisms, and to calculate automorphism groups. We summarize these uses of Theorem \ref{theorem:mlayerisom} in the following corollary.
\begin{corollary}
The following statements are true for all multilayer networks $M,M^\prime \in \mathcal{M}$ and nonempty $p$:
\begin{enumerate}
\item[(1)] $M \cong_p M' \iff f_p(M) \cong f_p(M')$;
\item[(2)] $C_G(f_p(M))$ is complete invariant for $\cong_p$ if $C_G$ is complete invariant for $\cong$;
\item[(3)] $\mathrm{Aut}_p(M)=g^{-1}_p(\mathrm{Aut}(f_p(M)))$.
\end{enumerate}
\label{corollary:misom1}
\end{corollary}
For a proof, see Section \ref{sec:proof_corollaries}.

We now define the ``multilayer network isomorphism decision problem'' and show that it is in the same complexity class with the graph isomorphism problem if one problem is allowed to be reduced to the other in polynomial time. 

\begin{definition}
The multilayer network isomorphism problem ($MGI_p$) gives a solution to the following decision problem: Given two multilayer networks $M,M^\prime \in \mathcal{M}$, is $M \cong_p M^\prime$ true? 
\end{definition}

The complexity class in which problems can be reduced to the graph isomorphism problem is denote here  $\mathit{GI}$ and many graph-related problems such as vertex-colored graph isomorphism problem and hypergraph isomorphism problem are known to be  $\mathit{GI}$-complete \cite{Zemlyachenko1985Graph}.

\begin{corollary}
$\mathit{MGI}_p$ is $\mathit{GI}$-complete for all nonempty $p$.
\label{corollary:complexity}
\end{corollary}
For a proof, see Section \ref{sec:proof_corollaries}.

We do the reduction from multilayer networks to vertex-colored graphs using the $f_p$ function that we defined earlier. We only need to show that this reduction is indeed linear (and thus also polynomial) in time. 
The reduction of graph isomorphism problems to multilayer network isomorphism problems is trivial if we allow the vertex labels to be permuted, because we can simply map the graph to a multilayer network with a single layer. If we cannot permute the vertex labels---i.e., if $0 \notin p$---then we need to construct a multilayer network in which each vertex of the graph becomes a layer with only a single vertex and we then connect these layers according to the graph adjacencies.


\section{Isomorphisms Induced for Other Types of Networks} \label{sec:examples}

In this section, we illustrate the use of multilayer network isomorphisms in network representations that can be mapped into the multilayer-network framework. As example, we use the three most common types of multilayer networks \cite{Kivela2014Multilayer}: multiplex networks, vertex-colored networks, and temporal networks. In Section~\ref{sec:multiplex}, we discuss isomorphisms in multiplex networks. We focus on counting the number of nonisomorphic multiplex networks of a given size (i.e., with a given number of vertices). In Section~\ref{sec:vertexcolored}, we discuss isomorphisms in vertex-colored networks. In Section~\ref{sec:temporal}, we illustrate how multilayer network isomorphisms give a natural definition of the isomorphisms that are defined implicitly for temporal networks when analyzing motifs in them~\cite{Kovanen2011Temporal}.


\subsection{Multiplex Networks} \label{sec:multiplex}

\begin{figure*}[!htp]
\includegraphics[width=0.45\linewidth]{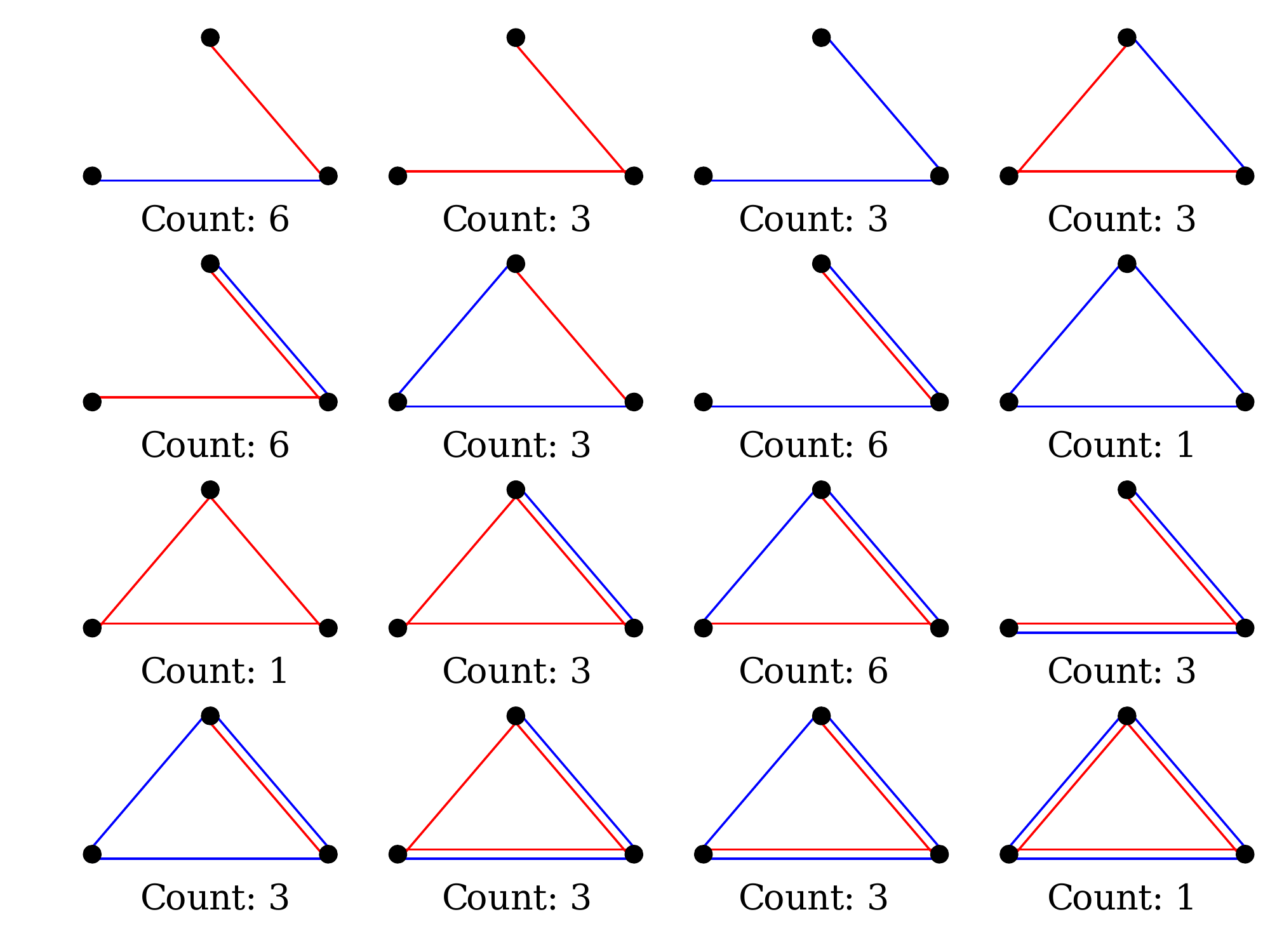}
\hspace{.2 in}
\includegraphics[width=0.45\linewidth]{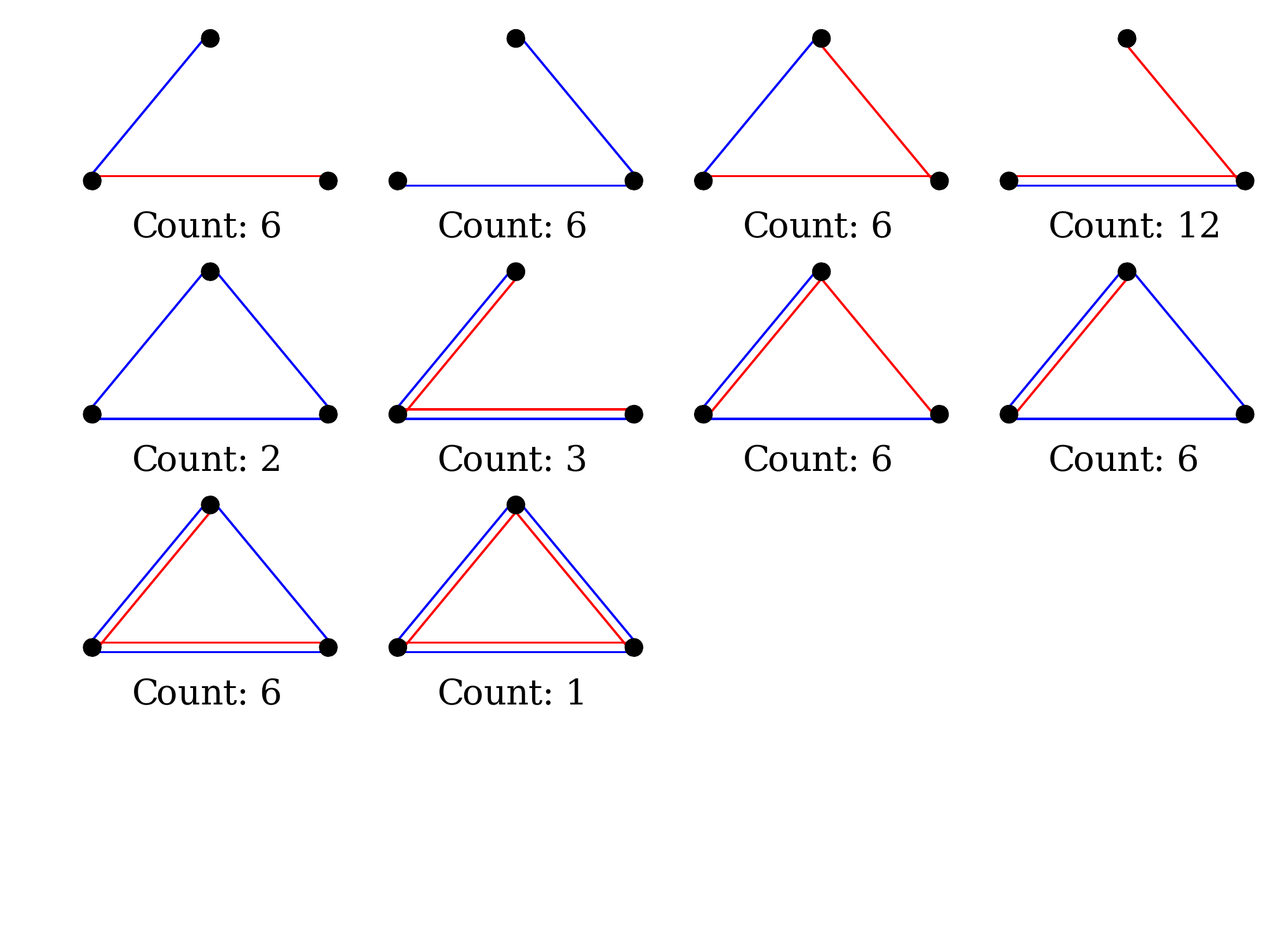}
\caption{Isomorphism classes for multiplex networks with 3 vertices and 2 layers. We only include connected networks. We show vertex isomorphisms in the left panel and vertex-layer isomorphisms in the right panel. The count is the number of networks (with a fixed set of vertices and layers) that are mapped to each class.}
\label{fig:mplex_3_2}
\end{figure*}

Multiplex networks have thus far been the most popular type of multilayer networks for analyzing empirical network data~\cite{Kivela2014Multilayer,Boccaletti2014}. One can represent systems that have several different types of interactions between its vertices as multiplex networks that are defined as a sequence of graphs $\{ G_\alpha \}_\alpha = \{ V_\alpha,E_\alpha \}_\alpha $. It is almost always assumed that the set of vertices is the same in all of the layers $V_\alpha=V_\beta$ for all $\alpha,\beta$ (although this is not a requirement), and multiplex networks that satisfy this condition are said to be ``vertex-aligned''~\cite{Kivela2014Multilayer}.

One can map multiplex networks to multilayer networks with a single aspect by considering each of the graphs $G_\alpha$ as an intra-layer network (i.e., a network in which the edges are placed inside of a single layer~\cite{Kivela2014Multilayer}). Optionally, one can add inter-layer edges (i.e., edges in which the two vertices are in different layers) by linking each vertex to its replicates in other layers. This is known as \emph{categorical coupling}. Either using categorical coupling or leaving out all of the inter-layer edges leads to same isomorphism relations for multiplex networks. However, for \emph{ordinal coupling}, in which only vertices in consecutive layers are adjacent to each other, the isomorphism classes can be different (see Section~\ref{sec:temporal}). A vertex isomorphism in multiplex networks allows the vertex labels to be permuted, but the types of edges are preserved. The layer isomorphism allows the types of edges to be permuted but only in a way that all of the edges of a particular type are mapped to a single other type.

Analyzing small substructures using clustering coefficients in social networks and other multiplex networks have recently gained attention~\cite{Barrett2012Taking,Brodka2010Method,Brodka2012Analysis,Criado2011Mathematical,Battiston2014Metrics,Cozzo2013Clustering}. Such structures have important (and fascinating) new features that go beyond clustering coefficients in ordinary graphs. Instead of there being only one type of triangle, there is very large number of different types of multiplex triangles and connected triplets of vertices. Such triadic structures have not been fully explored, though we discuss them in some detail in a recent paper \cite{Cozzo2013Clustering}. Moreover, one can study larger subgraphs and induced subgraphs of multiplex networks by extending the analysis of ``motifs'' in graphs~\cite{Milo2002Network} to multiplex networks. There has already been interest in motif analysis in gene-interaction networks with multiple types of interactions \cite{Taylor2007Network}, in food webs that can be represented using directed ordered networks~\cite{Paulau2015Motif}, and in brain networks with both anatomical and functional connections \cite{batt2016}. Methods based on counting the number of isomorphic subgraphs, such as motif analysis, work best if the number of isomorphism classes is relatively small.  Such analysis also necessitates the investigation of isomorphisms for their own sake, and they thereby serve as an important motivation (as well as an obvious future direction) for the present work. In Fig.~\ref{fig:mplex_3_2}, we enumerate all of the possible isomorphisms in connected multiplex networks with 3 vertices and 2 layers. We indicate each of the $16$ vertex-isomorphism classes and $10$ vertex-layer-isomorphism classes.

The problem of counting the nonisomorphic graphs that have some restrictions is known as the ``graph enumeration problem'' in graph theory, and such problems can be extended to multiplex networks (or multilayer networks in general) using the theory that we have introduced in the present paper. The number of undirected graphs with a fixed set of $n$ vertices is $2^{\binom{n}{2}}$, and the number of nonisomorphic graphs also grows very quickly with $n$. In multiplex networks, the analogous problem is to count the number of multiplex networks with $n$ vertices and $b$ layers. For vertex-aligned multiplex networks, the number of networks is $2^{b\binom{n}{2}}$. In Table~\ref{table:mplexcounts}, we show the number of nonisomorphic vertex-aligned multiplex networks for small values of $n$ and $b$ when considering vertex isomorphism or vertex-layer isomorphism. We produce the numbers in the table by systematically going through all of the networks of a certain size and categorizing them according to their isomorphism class\footnote{In practice, of course, we did reduce the search space by taking advantage of symmetries in the problem.}. The layer isomorphism problem for multiplex networks does not require one to solve the graph isomorphism problem, and it is easy to solve analytically. For layer isomorphisms, the number of nonisomorphic networks in a single-aspect vertex-aligned multiplex network is given by the formula $\binom{2^{\binom{n}{2}}+b-1}{b}$.

\begin{table*}
\begin{center}
\begin{tabular}{cc|cccc}
$P_{\{ 0,1 \} }$&\multicolumn{1}{r}{}&\multicolumn{4}{c}{Vertices}\\
 &&2  & 3  & 4   &5\\
\cline{2-6}
&1& 2 & 4  & 11  & 34\\
&2& 3 & 13 & 154 & 5466\\
\multirow{-3}{*}{\begin{sideways}Layers\end{sideways}}
&3& 4 & 36 & 2381& 1540146\\
\end{tabular}
\hspace{1.5cm}
\begin{tabular}{cc|cccc}
$P_{\{ 0 \} } $&\multicolumn{1}{r}{}&\multicolumn{4}{c}{Vertices}\\
 &&2  & 3  & 4   &5\\
\cline{2-6}
&1& 2 & 4  & 11    & 34\\
&2& 4 & 20 & 276   & 10688\\
\multirow{-3}{*}{\begin{sideways}Layers\end{sideways}}
&3& 8 & 120 & 12496& 9156288\\
\end{tabular}
\end{center}
\begin{center}
\begin{tabular}{ccc|c|cccccccccccccc}
&\multicolumn{1}{r}{}&\multicolumn{16}{c}{}\\
&\multicolumn{1}{r}{}&\multicolumn{16}{c}{Number of Edges}\\
 &n&l&Total&0      & 1     & 2     &3      &4      &5 &6  &7  &8  &9  &10 &11 &12 &\\
\cline{2-17}
 &3&1&4    &{\bf 1}&{\bf 1}&{\bf 1}&{\bf 1}&       &  &   &   &   &   &   &   &   &\\
 &3&2&13   &1      & 1     &{\bf 3}&{\bf 3}&{\bf 3}&1 &1  &   &   &   &   &   &   &\\
\cline{2-3}
 &4&1&11   &1      & 1     & 2     &{\bf 3}&2      &1 &1  &   &   &   &   &   &   &\\
 &4&2&154  &1      & 1     & 5     &9      &20     &24&{\bf 34} &24 &20 &9  &5  &1  &1  &\\
 &4&3&2381 &1      & 1     & 5     &15     &39     &88&178&280&375&{\bf 417}&375&280&178&\dots \\
\multirow{-4}{*}{\begin{sideways}$P_{\{ 0,1 \} }$\end{sideways}}
 &4&4&34797&1      & 1     & 5     &15     &50     &132&366&800&1619&2715&4005&4973&{\bf 5433}&\dots \\ 
\end{tabular}
\end{center}
\label{table:mplexcounts}
\caption{(Top) Number of isomorphism classes in multiplex networks for (left) vertex-layer isomorphisms and (right) vertex isomorphisms. (Bottom) The numbers of isomorphism classes with a given number of edges. All of the rows are symmetric around the maximum value(s), which we indicate in bold.
The isomorphism classes were enumerated using \protect\cite{pymnet}.
}
\end{table*}


\subsection{Vertex-Colored Networks} \label{sec:vertexcolored}

One can represent networks with multiple types (i.e., colors, labels, etc.) of vertices using the vertex-colored graphs that we discussed in Section~\ref{sec:vertexcoloredisom}. One can also map structures such as networks of networks, interconnected networks, and interdependent networks into the same class of multilayer networks~\cite{Kivela2014Multilayer}, because one can mark each subnetworks in any of these structures using a given vertex color.

One can map vertex-colored networks into multilayer networks by considering each color as a layer. One then adds vertices to the layer that corresponds to their color. Each vertex thus occurs in only a single layer, and one can add edges between the vertices in the resulting multilayer network exactly as they appear in the vertex-colored network. That is, in this multilayer-network representation, all inter-layer and intra-layer edges are possible.

Vertex isomorphisms in this case are the normal isomorphisms of vertex-colored graphs, as vertex labels can be permuted but the colors are left unchanged. In layer isomorphisms, the vertex labels must be left untouched, but the colors can be permuted. For example, consider two networks with the same topology but different colorings that correspond to vertex classifications (e.g., community assignments \cite{Porter2009}) of vertices. Two networks are then layer isomorphic if the two vertex classifications are the same. In a vertex-layer isomorphism, one can permute both the vertex names and the colors.


\subsection{Temporal Networks}
\label{sec:temporal}

\begin{figure*}[!htp]
\includegraphics[width=1.0\linewidth]{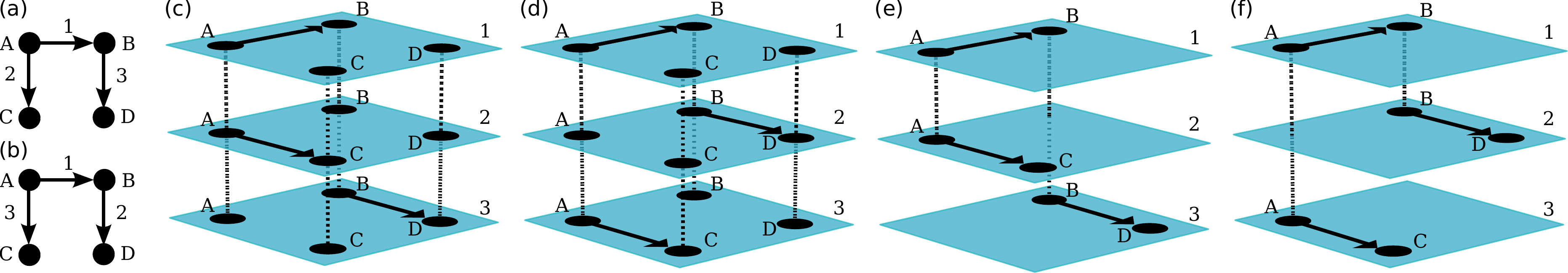}
\caption{ (a,b) Two event-based directed temporal networks that were used as an example in Ref.~\protect\cite{Kovanen2011Temporal} to illustrate the difference between temporal motifs and flow motifs. The two temporal networks correspond to two distinct temporal motifs (i.e., two distinct isomorphism classes) but to a same flow motif (i.e., the same isomorphism class). The numbers next to the edges are times at which events take place. (c,d) Representations of the two temporal networks as vertex-aligned multiplex networks in which each vertex is present on each layer and the layers are ordinally coupled. This representation leads to the same isomorphism as used for temporal motifs in Ref.~\protect\cite{Kovanen2011Temporal}, and the two multilayer networks are not isomorphic, because the coupling edges fully determine the relative order of all layers. (e,f) Representation of the two temporal networks as non-vertex-aligned multiplex networks. In this representation, vertices are only present on layers in which they are active, and they are only adjacent to their replicas in other layers that participate in events. Consequently, similar to the isomorphisms that were used to define flow motifs in Ref.~\protect\cite{Kovanen2011Temporal}, the relative order of events is only important for events that are adjacent. The two multilayer networks constructed in this way are thus vertex-layer isomorphic.
Multilayer-network illustrations produced using \protect\cite{pymnet}.
}
\label{fig:temporal}
\end{figure*}

Temporal networks in which each edge and vertex are present only at certain time instances arise in a large variety of scientific disciplines (e.g., sociology, cell biology, ecology, communication, infrastructure, and more)~\cite{Holme2012Temporal}.  (One can also think about temporal networks with intervals of activity or with continuous time.) One can represent such temporal networks as multilayer networks~\cite{DeDomenico2013Mathematical,Kivela2014Multilayer}, although this is not the usual framework that has been used to study them.  (See \cite{Mucha2010Community} for an early study that used this perspective.) Representing temporal networks as multilayer networks allows one to use ideas and methodology from the theory of multilayer networks to study them, and this has already been profitable in application areas such as political science \cite{Mucha2010Community}, neuroscience \cite{Bassett2011}, finance \cite{bazzi2015}, and sociology \cite{myers2015}.
More typically, one represents temporal networks either as contact sequences or time sequence of graphs~\cite{Holme2012Temporal}. Sequences of graphs are very similar construction to multiplex networks, where the key difference is that the order of the graphs in the sequence is important. One can map this type of temporal network to a multilayer network in very similar way as with multiplex networks. For temporal networks, however, one typically uses ordinal coupling instead of categorical coupling, although it is possible to be more general \cite{Kivela2014Multilayer}. (In other words, instead of coupling all of the layer together, one only couples consecutive layers~\cite{Mucha2010Community,Kivela2014Multilayer}.)

A contact sequence consists of a set of triplets $(u,v,t)$ that each represents a (possibly directed) contact between vertices $u$ and $v$ at time $t$. It is common to represent contact-sequence data as a sequence of very sparse graphs in which each distinct time stamp corresponds to a graph, and two vertices are adjacent in such a graph if they participate in an event at that time stamp~\cite{Holme2012Temporal}. This representation leads naturally to the multiplex-like multilayer network representation of contact sequences that we described above. Alternatively, one can represent each event as a layer that only includes the two (or potentially more) vertices that participate in the event. The two vertices in the layer are each adjacent to its replicas in temporally adjacent layers.  (See our earlier discussion of ordinal coupling.) These two alternative representations of temporal networks induce different isomorphism relations, and this difference is related to the difference between the temporal motifs and flow motifs from Ref.~\cite{Kovanen2011Temporal}. We illustrate this distinction using an example in Fig.~\ref{fig:temporal}. 

Contact sequences can also include delay or duration of the contact~\cite{Holme2012Temporal}. The delay (or latency) implies that the effect of a contact is not instantaneous. For example, in a temporal network of airline traffic, one can construe the flight time of each flight as a delay, and this can have an effect on the temporal paths and dynamical processes on the network~\cite{Pan2011Path}. This type of temporal network can also be represented using a multilayer-network framework~\cite{Kivela2014Multilayer}. For example, a flight that leaves city $A$ at time $t_1$ and arrives in city $B$ at time $t_2$ is represented as an edge from vertex $A$ in layer $t_1$ to vertex $B$ in layer $t_2$. Consequently, multilayer network isomorphisms can also be used for temporal networks with delays.

In a network that is purely temporal, and which thus has only a single aspect, there are three different possible multilayer isomorphisms. (1) Two temporal networks are vertex-isomorphic if they exhibit the same temporal patterns at exactly the same time but between (possibly) different vertices. (2) Two temporal networks are layer-isomorphic if they exhibit exactly the same temporal patterns with exactly the same vertices, although the actual times (though not the relative order of events) can change. (3)  Two temporal networks are vertex-layer isomorphic if they exhibit exactly the same temporal pattern, though possibly using different vertices, but the actual times (although not the relative order of events) can change.


\section{Conclusions and Discussion} \label{sec:conclusions}

The theory of multilayer network isomorphisms illustrates the power of the multilayer-network formalism: Any concept or method that can be defined for general multilayer networks immediately yields the same concept or method for any type of network that can be construed as a type of multilayer network. The interpretation of the concepts or methods depends on the application and scientific question of interest, but the underlying mathematics is the same. In this sense, multilayer networks allow one to return to the early 
days of network science in which simple graphs were used to represent myriad types of systems and the same tools could be applied to all of them. The key difference is that multilayer networks allow one to represent much richer and application-specific structures.

Going from graphs to multilayer networks adds a ``degree of freedom'' to ordinary networks (or multiple degrees of freedom if the number of aspects is larger than $1$), and generalizing concepts defined for graphs thus typically leads to multiple alternative definitions~\cite{Kivela2014Multilayer}. This is also true for graph isomorphisms and any isomorphism-based methods in multilayer networks, and this is underscores why it is important to identify multiple types of multilayer network isomorphisms. Given a problem under study, one still needs to decide which of these generalizations to use. Naturally, one can also examine multiple types of isomorphisms.

Our work on multilayer network isomorphisms lays the foundation for many future research directions in the study of multilayer networks. Motif analysis can now be generalized for any type of multilayer network once one defines a proper null model for the type of multilayer network under study. A good selection of network models already exist both for multiplex networks and for vertex-colored networks and similar structures~\cite{Kivela2014Multilayer}. Another straightforward application of isomorphisms in multilayer networks is the calculation of structural roles \cite{doreianbook,Borgatti1992Notions} by defining two vertices to be structurally equivalent if they are equivalent under an automorphism.  One can also examine other types of role equivalence in a multilayer setting.

One of the challenges in isomorphism-based analysis methods is that they are computationally challenging even for ordinary graphs. We introduced a computationally efficient way of deducing if two multilayer networks are isomorphic and calculating multilayer network certificates by reducing the problem to the isomorphism problem for vertex-colored graphs. Although this method is efficient for general multilayer networks, there is room for improvement when one is only considering a specific type of multilayer network (such as multiplex networks). 

Our theory also forms a basis for methods that still need some additional work to be generalized for multilayer networks. For example, in interesting direction would be to define ``approximate isomorphisms'' or inexact graph matching~\cite{Conte2004Thirty} along with a way to measure how close one is to achieving an isomorphism. This would, in turn, allow one to define similarity measures between multilayer networks and techniques for ``aligning'' two multilayer networks. It would also make it possible to relax the conditions in role equivalence to better study structural roles in multilayer networks.

Perhaps the most exciting direction in research on multilayer networks is the development of methods and models that are not direct generalizations of any of the traditional methods and models for ordinary graphs \cite{Kivela2014Multilayer}. The fact that there are multiple types of isomorphisms opens up the possibility to help develop such methodology by comparing different types of isomorphism classes.  We also believe that there will be an increasing need for the study of networks that have multiple aspects (e.g., both time-dependence and multiplexity), and our isomorphism framework is ready to be used for such networks. 


\section{Acknowledgements} \label{acks}

Both authors were supported by the European Commission FET-Proactive project PLEXMATH (Grant No. 317614). We thank Robert Gevorkyan and Puck Rombach for helpful comments.

\bibliography{mlayerisom}


\section{Proofs}

\subsection{Proofs  of basic properties of isomorphism and automorphism groups}
\label{sec:proofs}

\begin{proof}[Proof of Proposition \ref{prop:autoprop}]
(1) Take any $\bs\zeta \in \mathrm{Aut}_{p^\prime}(M)$. It follows that $M^{\bs\zeta}=M$ and $\bs\zeta \in P_p$ because $\bs\zeta \in P_{p^\prime}$. That is, $\bs\zeta \in \mathrm{Aut}_{p}(M)$.

(2) Both $\mathrm{Aut}_{p_1}(M)$ and $\mathrm{Aut}_{p_2}(M)$ are subgroups of $ \mathrm{Aut}_p(M)$ because of (1). Their direct product is a group if they commute. 
Take any $\bs\zeta \in \mathrm{Aut}_{p_1}(M)$ and $\bs\zeta^\prime \in \mathrm{Aut}_{p_2}(M)$. We have $(\bs\zeta\bs\zeta^\prime)_a=(\bs\zeta \bs{1})_a=(\bs{1} \bs\zeta)_a=(\bs\zeta^\prime\bs\zeta)_a$ if $a \in p_1$, $(\bs\zeta\bs\zeta^\prime)_a=(\bs{1} \bs\zeta^\prime)_a=(\bs\zeta^\prime \bs{1})_a=(\bs\zeta^\prime\bs\zeta)_a$ if $a \in p_2$, and  $(\bs\zeta\bs\zeta^\prime)_a=(\bs{1})_a=(\bs\zeta^\prime\bs\zeta)_a$ if $a \notin p_1,p_2$. Therefore, $\bs\zeta \bs\zeta^\prime = \bs\zeta^\prime \bs\zeta$ and $\mathrm{Aut}_{p_1}(M)\mathrm{Aut}_{p_2}(M)=\mathrm{Aut}_{p_2}(M)\mathrm{Aut}_{p_1}(M)$.

(3) Let us look at arbitrary aspect $a$. Because $p_i \cap p_j = \emptyset$ for all $i \neq j$, it follows that $a$ is either a member of exactly one $p_i$ or of none of them. If $a$ is not in any $p_i$, then $D_a^{p_i}=\mathbbm{1}_{L_a}$ and $\zeta^{(i)}_a=1_{L_a}$ for all $i$. However, if $a \in p_j$ (i.e., the aspect $a$ is in exactly one set), then $(\prod_i \bs\zeta^{(i)})_a=(\bs\zeta^{(j)})_a$, and it thus follows that $(\bs\zeta^{(j)})_a=1_{L_a}$. Because $a$ is arbitrary, we have shown that $\bs\zeta^{(i)}=\bs{1}$ for all $i$.
\end{proof}

\begin{proof}[Proof of Proposition \ref{proposition:aspectpermutation}]
We first show that $\mathrm{Iso}_{p}(M,M^\prime) \subseteq I_{\sigma^{-1}} [\mathrm{Iso}_{p^{\sigma}}(A_\sigma(M),A_\sigma(M^\prime))]$. We consider any $\bs\zeta \in \mathrm{Iso}_p(M,M^\prime)$ and show that $I_\sigma(\bs\zeta) \in \mathrm{Iso}_{p^{\sigma}}(A_\sigma(M),A_\sigma(M^\prime))$. By a direct calculation, $A_\sigma(M)^{I_\sigma(\bs\zeta)}=A_\sigma(M^\prime)$: for vertex-layer tuples, $(I_\sigma(V_M))^{I_\sigma(\bs\zeta)}=\{ v_{\sigma^{-1}(0)}^{\zeta_{\sigma^{-1}(0)}},\dots,v_{\sigma^{-1}(d)}^{\zeta_{\sigma^{-1}(d)}} \mid {\bf v} \in V_M\} = I_\sigma(V_M^{\bs\zeta})=I_\sigma(V_M^\prime)$; for edges, $\{ (I_\sigma({\bf v}),I_\sigma({\bf u})) \mid ({\bf v},{\bf u}) \in E_M\}^{I_\sigma(\bs\zeta)} = \{ (I_\sigma({\bf v}^{\bs\zeta}),I_\sigma({\bf u}^{\bs\zeta})) \mid ({\bf v},{\bf u}) \in E_M\} = \{ (I_\sigma({\bf v}),I_\sigma({\bf u})) \mid ({\bf v},{\bf u}) \in E_M^\prime\}$; for vertices, $L_{\sigma^{-1}(0)}^{\zeta_{\sigma^{-1}(0)}} =L_{\sigma^{-1}(0)}^\prime$; and for elementary layers, $( \{ L_{\sigma^{-1}(a)}\}_{a=1}^d)^{I_\sigma(\bs\zeta)}= \{ L_{\sigma^{-1}(a)}^{\zeta_{\sigma^{-1}(a)}}\}_{a=1}^d=\{ L_{\sigma^{-1}(a)}^\prime\}_{a=1}^d$, because $L_a^{\zeta_a}=L_a^\prime$ for all $a$. 

Now we need to show that $I_\sigma(\bs\zeta)$ is an acceptable mapping for the isomorphism on the right-hand side of equation (\ref{aspectpermrelation}). Note that the definition of $P_p$ in equation (\ref{eq:ppdef}) depends on the sets $\{ L_a \}_0^d$of elementary layers, and these sets are different in the two isomorphisms in the two sides of equation (\ref{aspectpermrelation}). We write this dependency explicitly, so that $P_p$ in the left isomorphism becomes $P_p(\{ L_a \}_0^d)$ and $P_p^{\sigma}$ in the right isomorphism becomes $P_{p^{\sigma}}(\{ L_{\sigma^{-1}(a)} \}_0^d)$. With this notation, $I_\sigma(P_p(\{ L_a \}_0^d))=P_{p^{\sigma}}(\{ L_{\sigma^{-1}(a)} \}_0^d)$, so $\bs\zeta \in P_p(\{ L_a \}_0^d) \implies I_\sigma(\bs\zeta) \in P_{p^{\sigma}}(\{ L_{\sigma^{-1}(a)} \}_0^d)$.

Now that we know that $\mathrm{Iso}_{p}(M,M^\prime) \subseteq I_{\sigma^{-1}} [\mathrm{Iso}_{p^{\sigma}}(A_\sigma(M),A_\sigma(M^\prime))]$ for any aspect permutation $\sigma$, we can use the aspect permutation $\sigma^{-1}$ instead of $\sigma$. Consequently, we can write 
$	I_{\sigma^{-1}} [\mathrm{Iso}_{p^{\sigma}}(A_\sigma(M),A_\sigma(M^\prime))] \subseteq I_{\sigma^{-1}} [I_{\sigma} [\mathrm{Iso}_{(p^{\sigma})^{\sigma^{-1}}}(A_{\sigma^{-1}}(A_\sigma(M)),A_{\sigma^{-1}}(A_\sigma(M^\prime)))]]  =\mathrm{Iso}_{p}(M,M^\prime)\,.
$
\end{proof}


\subsection{Proof of the Reduction Theorem}
\label{sec:proof_reduction}

We will need the following lemma for our proof of Theorem~\ref{theorem:mlayerisom}.
\begin{lemma}
\label{lemma:mapping}
Suppose that $f:\mathcal{M} \to \mathcal{G}_c$ and $g$ maps permutations $P_p$ of $M \in \mathcal{M}$ to permutations of $G_c \in f(\mathcal{M})$. In addition, we suppose that the following conditions hold:
\begin{enumerate}
\item[(1)] $f$ and $g$ are injective;
\item[(2)] $f(M)^\gamma=f(M^\prime) \implies \gamma \in g(P_p)$;
\item[(3)] for all $\zeta \in P_p$, we have $f(M^\zeta)=f(M)^{g(\zeta)}$.
\end{enumerate}
It then follows that $\mathrm{Iso}_p(M,M^\prime)=g^{-1}(\mathrm{Iso}(f(M),f(M^\prime)))$.
\end{lemma}

\begin{proof}[Proof of Lemma \ref{lemma:mapping}]
Take any $\pi \in \mathrm{Iso}_p(M,M^\prime)$. Because of condition (3), it then follows that $f(M)^{g(\pi)}=f(M^\pi)=f(M^\prime)$ and thus that $g(\pi) \in \mathrm{Iso}(f(M),f(M^\prime))$. This gives $\pi \in g^{-1}(\mathrm{Iso}(f(M),f(M^\prime)))$ and $\mathrm{Iso}_p(M,M^\prime) \subseteq g^{-1}(\mathrm{Iso}(f(M),f(M^\prime)))$.
Now let $\gamma \in \mathrm{Iso}(f(M),f(M^\prime))$. Because of condition (2), $\gamma \in g(P_p)$ and $g^{-1}(\gamma) \in P_p$. Using (3), we can then write that  $M^\prime=f^{-1}(f(M^\prime))=f^{-1}(f(M)^\gamma)=f^{-1}(f(M)^{g(g^{-1}(\gamma))})=f^{-1}(f(M^{g^{-1}(\gamma)}))=M^{g^{-1}(\gamma)}$. Thus, $g^{-1}(\gamma) \in \mathrm{Iso}_p(M,M^\prime)$, which implies that $g^{-1}(\mathrm{Iso}(f(M),f(M^\prime))) \subseteq \mathrm{Iso}_p(M,M^\prime)$. Consequently, $\mathrm{Iso}_p(M,M^\prime)=g^{-1}(Iso(f(M),f(M^\prime)))$. 
\end{proof}

\begin{proof}[Proof of Theorem \ref{theorem:mlayerisom}]

We now use Lemma~\ref{lemma:mapping} to prove Theorem \ref{theorem:mlayerisom}. We prove each of the three conditions for $f_p$ and $g_p$ that we need to apply Lemma~\ref{lemma:mapping}.

\vspace{.1 in}

We begin by proving condition (1). 

First, we show that $g_p$ is injective. Take any $\bs\zeta,\bs\zeta^\prime \in P_p$ such that $g_p(\bs\zeta)=g_p(\bs\zeta^\prime)$. For any $a \notin p$, it follows by definition of $P_p$ that $\zeta_a = 1_{L_a} = \zeta^\prime_a$, where $1_{L_a}$ is an identity permutation over the set $L_a$. For $a \in p$, the definition of $g_p$ guarantees that $v^{\zeta_a} = v^{g_p(\bs\zeta)} = v^{g_p(\bs\zeta^\prime)}=v^{\zeta^\prime_a}$ for all $v \in L_a$. That is, $\bs\zeta=\bs\zeta^\prime$, so $g_p$ is injective.

We now prove that $f_p$ is injective. Take any $M, M^\prime \in \mathcal{M}$ such that $f_p(M)=f_p(M^\prime)$. It follows that $V_M \cup V_0 = V_M^\prime \cup V_0^\prime$, $E_M \cup E_0 = E_M^\prime \cup E_0^\prime$, 
and $p \cup L_{\overline{p}_1} \times \dots \times L_{\overline{p}_m} = p \cup L_{\overline{p}_1}^\prime \times \dots \times L_{\overline{p}_m}^\prime$. Because we assumed that there are no shared labels of vertices or elementary layers (and that tuples of elementary layers and vertices are not in the vertex set or in any elementary layer set), it follows that $V_M = V_M^\prime$ and $E_M = E_M^\prime$. Because $M, M^\prime \in \mathcal{M}$, it is also true that ${\bf L}={\bf L}^\prime$ and $V=V^\prime$. Thus, $M=M^\prime$ and $f_p$ is injective.

\vspace{.1 in}

We now prove condition (2). 

Consider an arbitrary $\gamma \in S_{V_{f_p(M)}}$ such that $f_p(M)^{\gamma}=f_p(M^{\prime})$. We want to construct $\bs\zeta \in P_p$ so that $g_p(\bs\zeta)=\gamma$. For any $v_a \in L_a$, we let $v_a^{\zeta_a}=v_a^{\gamma}$ if $a \in p$ and $v_a^{\zeta_a}=v_a$ if $a \in \overline{p}$. The $\bs\zeta$ defined in this way is in $P_p$ because permutations for $a \in \overline{p}$ are identity permutations and $v_a \in L_a$ yields $\pi(v_a)=a$ and thus $v_a^{\gamma} \in L_a$. 
We now have by definition that $v_a^{g_p(\bs\zeta)}=v_a^{\zeta_a}=v_a^{\gamma}$ for $v_a \in V_0$ and ${\bf v}^{g_p(\bs\zeta)}= {\bf v}^{\bs\zeta}$ for ${\bf v} \in V_M$.
If we assume that ${\bf v}^{\bs\zeta} \neq {\bf v}^{\gamma}$ for ${\bf v} \in V_M$, then there exists an $a$ such that $v_a^{\zeta_a} \neq ({\bf v}^{\gamma})_a$.
We know that $v_a^{\zeta_a}=v_a=({\bf v}^\gamma)_a$ for $a \in \overline{p}$ because of the coloring: $(v_{\overline{p}_1},\dots,v_{\overline{p}_m})=\pi({\bf v})=\pi(\gamma^{-1}({\bf v}^\gamma))=\pi^\gamma({\bf v}^\gamma) =\pi^\prime({\bf v}^\gamma)=[({\bf v}^\gamma)_{\overline{p}_1},\dots,({\bf v}^\gamma)_{\overline{p}_m}]$.
That is, it must be true that $v_a^{\zeta_a} \neq ({\bf v}^{\gamma})_a$ for $a \in p$. 
Because $M^\prime$ is constructed using the function $f_p$, we know that $(u,{\bf v}^\gamma) \in E_0^\prime$ guarantees that there exists a $b \in p$ such that $u=({\bf v}^\gamma)_b$. However, $(v_a,{\bf v})^\gamma=(v_a^\gamma,{\bf v}^\gamma)=(v_a^{\zeta_a},{\bf v}^\gamma) \in E_0^\gamma = E_0^\prime$. Thus, there is a $b \in p$ so that $v_a^{\zeta_a}=({\bf v}^\gamma)_b$, and it thus follows that $a \neq b$. This is a contradiction, because $L_a \cap L_b = \emptyset$, and it thus must be true that $v_a^{\zeta_a} = ({\bf v}^{\gamma})_a$ for all ${\bf v} \in V_M$.

\vspace{.1 in}

We now prove condition (3).

From a direct calculation, we verify that for all $\bs\zeta \in P_p$, we have $f_p(M^{\bs\zeta})=f_p(M)^{g_p(\bs\zeta)}$. 

For vertices, we write $V_M^{\bs\zeta}=V_M^{g_p(\bs\zeta)}$ and $\bigcup_{a \in p} L_a^{\zeta_a}=\bigcup_{a \in p} L_a^{g_p(\bs\zeta)}=V_0^{g_p(\bs\zeta)}$. Combining these two equations yields $V_M^{g_p(\bs\zeta)} \cup V_0^{g_p(\bs\zeta)}= (V_M \cup V_0)^{g_p(\bs\zeta)}  = V_G^{g_p(\bs\zeta)}$.

For edges, we write $E_M^{\bs\zeta}=E_M^{g_p(\bs\zeta)}$ because $E_M \subset V_M \times V_M$, and it is also true that $E_0^{\bs\zeta}=\{ (v_a^{\zeta_a},{\bf v}^{\bs\zeta}) \mid {\bf v} \in V_M, a \in p \} =  \{ (v_a,{\bf v}) \mid {\bf v} \in V_M, a \in p \}^{g_p(\bs\zeta)}=E_0^{g_p(\bs\zeta)}$. 
Combining the two equations yields $E_M^{\bs\zeta} \cup E_0^{\bs\zeta} = E_M^{g_p(\bs\zeta)} \cup E_0^{g_p(\bs\zeta)}=(E_M \cup E_0)^{g_p(\bs\zeta)}=E_G^{g_p(\bs\zeta)}$.

For the color set $C$, the permutation $\bs\zeta \in P_p$ does not change anything because it only permutes the aspects in $p$. Additionally, the permutation $g_p(\bs\zeta)$ of the vertex-colored graph does not change any vertex colors by definition.

The color map $\pi^{g_p(\bs\zeta)}({\bf v})=\pi([g_p(\bs\zeta)]^{-1}({\bf v}))=\pi(\bs\zeta^{-1}({\bf v}))=(v_{\overline{p}_1} \dots v_{\overline{p}_m})=\pi({\bf v})$ if $v_g={\bf v} \in V_M$, where the third equality is true because $\bs\zeta^{-1}({\bf v}) \in V_M$. 
Similarly, $\pi^{g_p(\bs\zeta)}(v_a)=\pi([g_p(\bs\zeta)]^{-1}(v_a))=\pi(\zeta_a^{-1}(v_a))=a=\pi(v_a)$ if $v_a \in L_a$, where the third equality is true because $\zeta_a^{-1}(v_a) \in L_a$.
\end{proof}


\subsection{Proof of Corollaries}
\label{sec:proof_corollaries}

\begin{proof}[Proof of Corollary \ref{corollary:misom1}]

These results follow immediately from Theorem \ref{theorem:mlayerisom}.

\vspace{.1 in}

(1):  $M \cong_p M' \iff \mathrm{Iso}_p(M,M^\prime) \neq \emptyset \iff g^{-1}_p[\mathrm{Iso}(f_p(M),f_p(M^\prime))] \neq \emptyset \iff f_p(M) \cong f_p(M')$.

\vspace{.1 in}

(2): Let $C$ be the complete invariant of $\cong$ for vertex-colored graphs. That is, $C(G)=C(G^\prime) \iff G \cong G^\prime$, where $G,G^\prime \in \mathcal{G_C}$. From this invariance and (1), it follows that $C(f(M))=C(f(M^\prime)) \iff f(M) \cong f(M^\prime) \iff M \cong^* M^\prime $.

\vspace{.1 in}

(3): To obtain this result, we let $M^\prime=M$ in Theorem \ref{theorem:mlayerisom}.
\end{proof}

\begin{proof}[Proof of Corollary \ref{corollary:complexity}]
The number of vertices in $f_p(M)$ (see Definition~\ref{definition:fp}) is $|V_M|+\sum_a^p |L_a|$, the number of edges is $|E_M|+|V_M||p|$, and the number of colors can be limited to the number of vertices. In the function $f_p$, constructing each vertex, edge, or vertex color consists of copying it directly from the multilayer network or doing several operations of checking if an element belongs to a set that grows polynomially with the size of $M$. 
Thus, one can use point (1) in Corollary~\ref{corollary:misom1} to create a reduction that is polynomial in time (and linear in space) from $\mathit{MGI}_p$ to the vertex-colored graph isomorphism problem, which is known to be in $\mathit{GI}$. 
One can reduce in polynomial time any problem in $\mathit{GI}$ to $\mathit{MGI}_p$ by mapping the two graphs to the following multilayer networks. Choose $a \in p$ and use the set of vertices in the graph as a set of elementary layers in the aspect $a$. For the aspects $b \neq a$, add a single layer $l_b$ to the remaining elementary layer sets. For each vertex $u \in V$ in the graph, create a single vertex-layer ${\bf v}$ such that $v_a=u$ and $v_b=l_b$. (In other words, create a vertex ${\bf v}= (v_1,\dots,v_a,\dots,v_d)=(l_1,\dots,u,\dots,l_d)$.)  For each edge $(u,w) \in E$ in the graph, add an edge $[(l_1,\dots,u,\dots,l_d),(l_1,\dots,w,\dots,l_d)]$ to the multilayer network. The two multilayer networks are isomorphic according to $\cong_p$ exactly when the two graphs are isomorphic.
\end{proof}

\end{document}